\documentclass[%
reprint,
superscriptaddress,
frontmatterverbose, 
preprintnumbers,
longbibliography,
nofootinbib,
amsmath,
amssymb,
aps,
pra,
notitlepage,
]{revtex4-1}

\usepackage[titletoc,title]{appendix}
\usepackage[dvipsnames]{xcolor}
\usepackage{color}
\definecolor{darkRed}{rgb}{0.8, 0.0, 0.0}
\definecolor{darkBlue}{rgb}{0, 0.0, 0.8}
\definecolor{Gray}{gray}{0.92}
\definecolor{Gray2}{gray}{0.75}
\definecolor{maroon}{cmyk}{0,0.87,0.68,0.32}

\usepackage{hyperref}
\hypersetup{colorlinks=true,citecolor=darkBlue,linkcolor=darkBlue,filecolor=darkBlue,urlcolor=darkBlue,breaklinks=true}

\usepackage[marginal]{footmisc}
\usepackage{url}

\usepackage{mathtools}
\usepackage{amsmath}
\usepackage{amssymb}
\usepackage[shortlabels]{enumitem}
\usepackage{graphicx,epic,eepic,epsfig,latexsym,verbatim}
\usepackage{dsfont}

\usepackage{mathrsfs}

\usepackage{framed}
\definecolor{shadecolor}{rgb}{0.9,0.9,0.9}
\usepackage{float}
\usepackage{tikz}
\usepackage{tcolorbox}
\usepackage{relsize}

\usepackage{url}

\usepackage{breakurl}

\usepackage{amsthm}
\newtheorem{definition}{Definition}
\newtheorem{proposition}[definition]{Proposition}
\newtheorem{lemma}[definition]{Lemma}
\newtheorem{theorem}[definition]{Theorem}
\newtheorem{corollary}[definition]{Corollary}

\def\squareforqed{\hbox{\rlap{$\sqcap$}$\sqcup$}}
\def\qed{\ifmmode\squareforqed\else{\unskip\nobreak\hfil
\penalty50\hskip1em\null\nobreak\hfil\squareforqed
\parfillskip=0pt\finalhyphendemerits=0\endgraf}\fi}
\def\endenv{\ifmmode\;\else{\unskip\nobreak\hfil
\penalty50\hskip1em\null\nobreak\hfil\;
\parfillskip=0pt\finalhyphendemerits=0\endgraf}\fi}

\theoremstyle{definition}

\newtheorem{exmp}{Example}

\newcounter{remark}
\newenvironment{remark}[1][]{\refstepcounter{remark}\par\medskip\noindent%
\textbf{Remark~\theremark #1.} }{\medskip}

\newcounter{example}

\mathchardef\ordinarycolon\mathcode`\:
\mathcode`\:=\string"8000
\def\vcentcolon{\mathrel{\mathop\ordinarycolon}}
\begingroup \catcode`\:=\active
  \lowercase{\endgroup
  \let :\vcentcolon
  }

\usepackage{colortbl}
\usepackage{pifont}
\usepackage{booktabs}
\usepackage{makecell}
\usepackage{diagbox}
\usepackage{multirow}

\newcommand{\nc}{\newcommand}
\nc{\rnc}{\renewcommand}
\nc{\bra}[1]{\langle#1|}
\nc{\ket}[1]{|#1\rangle}
\nc{\<}{\langle}
\rnc{\>}{\rangle}
\nc{\ketbra}[2]{|#1\rangle\!\langle#2|}
\nc{\braket}[2]{\langle#1|#2\rangle}
\nc{\braandket}[3]{\langle #1|#2|#3\rangle}
\nc{\proj}[1]{| #1\rangle\!\langle #1 |}
\nc{\avg}[1]{\langle#1\rangle}
\nc{\Rank}{\operatorname{Rank}}
\nc{\smfrac}[2]{\mbox{$\frac{#1}{#2}$}}
\nc{\tr}{\operatorname{Tr}}
\nc{\ox}{\otimes}

\nc{\cA}{{\cal A}}
\nc{\cB}{{\cal B}}
\nc{\cC}{{\cal C}}
\nc{\cD}{{\cal D}}
\nc{\cE}{{\cal E}}
\nc{\cF}{\mathscr{ F}}
\nc{\cG}{{\cal G}}
\nc{\cH}{{\cal H}}
\nc{\cI}{{\rm id}}
\nc{\cJ}{{\cal J}}
\nc{\cK}{{\cal K}}
\nc{\cL}{{\cal L}}
\nc{\cM}{{\cal M}}
\nc{\cN}{{\cal N}}
\nc{\cO}{{\cal O}}
\nc{\cP}{{\cal P}}
\nc{\cQ}{{\cal Q}}
\nc{\cR}{{\cal R}}
\nc{\cS}{{\cal S}}
\nc{\cT}{{\cal T}}
\nc{\cV}{{\cal V}}
\nc{\cU}{{\cal U}}
\nc{\cX}{{\cal X}}
\nc{\cY}{{\cal Y}}
\nc{\cZ}{{\cal Z}}
\nc{\cW}{{\cal W}}

\nc{\RR}{{{\mathbb R}}}
\nc{\CC}{{{\mathbb C}}}
\nc{\DD}{{{\mathbb D}}}
\nc{\FF}{{{\mathbb F}}}
\nc{\NN}{{{\mathbb N}}}
\nc{\ZZ}{{{\mathbb Z}}}
\nc{\PP}{{{\mathbb P}}}
\nc{\QQ}{{{\mathbb Q}}}
\nc{\UU}{{{\mathbb U}}}
\nc{\EE}{{{\mathbb E}}}
\nc{\id}{{\operatorname{id}}}

\nc{\1}{{\mathds{1}}}
\nc{\supp}{{\operatorname{supp}}}

\def\ve{\varepsilon}

\makeatletter
\def\grd@save@target#1{%
  \def\grd@target{#1}}
\def\grd@save@start#1{%
  \def\grd@start{#1}}
\tikzset{
  grid with coordinates/.style={
    to path={%
      \pgfextra{%
        \edef\grd@@target{(\tikztotarget)}%
        \tikz@scan@one@point\grd@save@target\grd@@target\relax
        \edef\grd@@start{(\tikztostart)}%
        \tikz@scan@one@point\grd@save@start\grd@@start\relax
        \draw[minor help lines,magenta] (\tikztostart) grid (\tikztotarget);
        \draw[major help lines] (\tikztostart) grid (\tikztotarget);
        \grd@start
        \pgfmathsetmacro{\grd@xa}{\the\pgf@x/1cm}
        \pgfmathsetmacro{\grd@ya}{\the\pgf@y/1cm}
        \grd@target
        \pgfmathsetmacro{\grd@xb}{\the\pgf@x/1cm}
        \pgfmathsetmacro{\grd@yb}{\the\pgf@y/1cm}
        \pgfmathsetmacro{\grd@xc}{\grd@xa + \pgfkeysvalueof{/tikz/grid with coordinates/major step}}
        \pgfmathsetmacro{\grd@yc}{\grd@ya + \pgfkeysvalueof{/tikz/grid with coordinates/major step}}
        \foreach \x in {\grd@xa,\grd@xc,...,\grd@xb}
        \node[anchor=north] at (\x,\grd@ya) {\pgfmathprintnumber{\x}};
        \foreach \y in {\grd@ya,\grd@yc,...,\grd@yb}
        \node[anchor=east] at (\grd@xa,\y) {\pgfmathprintnumber{\y}};
      }
    }
  },
  minor help lines/.style={
    help lines,
    step=\pgfkeysvalueof{/tikz/grid with coordinates/minor step}
  },
  major help lines/.style={
    help lines,
    line width=\pgfkeysvalueof{/tikz/grid with coordinates/major line width},
    step=\pgfkeysvalueof{/tikz/grid with coordinates/major step}
  },
  grid with coordinates/.cd,
  minor step/.initial=.2,
  major step/.initial=1,
  major line width/.initial=2pt,
}
\makeatother

\makeatletter
\newcommand*\rel@kern[1]{\kern#1\dimexpr\macc@kerna}
\newcommand*\widebar[1]{%
  \begingroup
  \def\mathaccent##1##2{%
    \rel@kern{0.8}%
    \overline{\rel@kern{-0.8}\macc@nucleus\rel@kern{0.2}}%
    \rel@kern{-0.2}%
  }%
  \macc@depth\@ne
  \let\math@bgroup\@empty \let\math@egroup\macc@set@skewchar
  \mathsurround\z@ \frozen@everymath{\mathgroup\macc@group\relax}%
  \macc@set@skewchar\relax
  \let\mathaccentV\macc@nested@a
  \macc@nested@a\relax111{#1}%
  \endgroup
}
\makeatother

\newcommand{\Renyi}{R\'{e}nyi }

\definecolor{colorone}{rgb}{1,0.36,0.03}
\definecolor{colortwo}{rgb}{0.54,0.71,0.03}
\definecolor{colorthree}{rgb}{0.01,0.51,0.93}
\definecolor{colorfour}{rgb}{0.47,0.26,0.58}
\definecolor{lightblue}{rgb}{0.01,0.51,0.93}

\usepackage{tikz}
\usepackage{hyperref}
\hypersetup{colorlinks=true,citecolor=blue,linkcolor=blue,filecolor=blue,urlcolor=blue,breaklinks=true}

\makeatletter
\def\grd@save@target#1{%
  \def\grd@target{#1}}
\def\grd@save@start#1{%
  \def\grd@start{#1}}
\tikzset{
  grid with coordinates/.style={
    to path={%
      \pgfextra{%
        \edef\grd@@target{(\tikztotarget)}%
        \tikz@scan@one@point\grd@save@target\grd@@target\relax
        \edef\grd@@start{(\tikztostart)}%
        \tikz@scan@one@point\grd@save@start\grd@@start\relax
        \draw[minor help lines,magenta] (\tikztostart) grid (\tikztotarget);
        \draw[major help lines] (\tikztostart) grid (\tikztotarget);
        \grd@start
        \pgfmathsetmacro{\grd@xa}{\the\pgf@x/1cm}
        \pgfmathsetmacro{\grd@ya}{\the\pgf@y/1cm}
        \grd@target
        \pgfmathsetmacro{\grd@xb}{\the\pgf@x/1cm}
        \pgfmathsetmacro{\grd@yb}{\the\pgf@y/1cm}
        \pgfmathsetmacro{\grd@xc}{\grd@xa + \pgfkeysvalueof{/tikz/grid with coordinates/major step}}
        \pgfmathsetmacro{\grd@yc}{\grd@ya + \pgfkeysvalueof{/tikz/grid with coordinates/major step}}
        \foreach \x in {\grd@xa,\grd@xc,...,\grd@xb}
        \node[anchor=north] at (\x,\grd@ya) {\pgfmathprintnumber{\x}};
        \foreach \y in {\grd@ya,\grd@yc,...,\grd@yb}
        \node[anchor=east] at (\grd@xa,\y) {\pgfmathprintnumber{\y}};
      }
    }
  },
  minor help lines/.style={
    help lines,
    step=\pgfkeysvalueof{/tikz/grid with coordinates/minor step}
  },
  major help lines/.style={
    help lines,
    line width=\pgfkeysvalueof{/tikz/grid with coordinates/major line width},
    step=\pgfkeysvalueof{/tikz/grid with coordinates/major step}
  },
  grid with coordinates/.cd,
  minor step/.initial=.2,
  major step/.initial=1,
  major line width/.initial=2pt,
}
\makeatother

\nc{\lfc}{\Gamma}
\nc{\wst}{\text{\rm W}}
\nc{\cho}{\text{\rm C}}
\nc{\ave}{\text{\rm A}}
\nc{\CCZ}{\text{\rm CCZ}}
\nc{\conv}{\text{\rm conv}}

\usepackage{comment}

\newcommand{\diamnorm}[1]{\left\| #1 \right\|_\diamond}

\newcommand{\zw}[1]{ { \color{red}  (ZW: #1) }}
\newcommand{\kf}[1]{ { \color{blue}  (KF: #1) }}

\newcommand{\new}[1]{{\color{black}#1}}

\begin{document}

\title{No-go theorems for quantum resource purification II: \\ new approach and channel theory}

\author{Kun Fang}
\email{fangkun02@baidu.com}
\affiliation{Institute for Quantum Computing, Baidu Research, Beijing 100193, China}
\affiliation{Institute for Quantum Computing, University of Waterloo, Waterloo, Ontario N2L 3G1, Canada}
\author{Zi-Wen Liu}
\email{zliu1@perimeterinstitute.ca}
\affiliation{Perimeter Institute for Theoretical Physics, Waterloo, Ontario N2L 2Y5, Canada}
\date{\today}

\begin{abstract}
It has been recently shown that there exist universal fundamental limits to the accuracy and efficiency of the transformation from noisy resource states to pure ones (e.g.,~distillation) in any well-behaved quantum resource theory [Fang/Liu, \href{https://journals.aps.org/prl/abstract/10.1103/PhysRevLett.125.060405}{Phys.\ Rev.\ Lett.\ {\bf 125}, 060405 (2020)}]. 
Here, we develop a novel and powerful method for analyzing the limitations on quantum resource purification, which not only leads to improved bounds that rule out exact purification for a broader range of noisy states and are tight in certain cases, but also enable us to establish a robust no-purification theory for quantum channel (dynamical) resources.  More specifically, we employ the new method to derive universal bounds on the error and cost of transforming generic noisy channels (where multiple instances can be used adaptively, in contrast to the state theory) to some unitary resource channel under any free channel-to-channel map. 
We address several cases of practical interest in more concrete terms, and discuss the connections and applications of our general results to distillation, quantum error correction, quantum Shannon theory, and quantum circuit synthesis.
\end{abstract}

\maketitle

\section{Introduction}

Quantum technologies, such as quantum computing, quantum communication, and quantum cryptography, are an exciting frontier of science, due to their promising potential of achieving substantial advantages over conventional methods that may spark an important technological revolution. 
However, quantum systems are inherently highly susceptible to noise and errors in real-world scenarios, which often make them unreliable or difficult to scale up.
This poses a serious challenge to realizing the potential power of quantum technologies in practice.
The noise problem is particularly pressing at the moment, as we are now at a critical juncture where we are starting to make real effort to put the theoretically blueprinted quantum technologies into practice \cite{Preskill2018quantumcomputingin,supremacy19}.
In order to ease the effects of noise, we would generally need techniques that can ``purify'' the noisy systems.
To this end, methods such as quantum error correction \cite{nielsen2010quantum} and distillation \cite{BBPS96:ent_dist,BBPSSW96:ent_dist,BDSW96:ent_qec,Bravyi2005} are developed and have become central research topics in quantum information.


Behind the power of quantum technologies is the manipulation and utilization of various forms of quantum ``resources'' such as  entanglement \cite{Horodecki:entanglement}, coherence \cite{coherenceRMP}, and ``magic'' \cite{Bravyi2005,Veitch_2014}.
These different kinds a quantum resources can be commonly understood and characterized using the universal framework of ``quantum resource theory'' (see, e.g.,~Ref.~\cite{ChitambarGour19} for an introduction), which have been under active developments in recent years.
Recently, Ref.~\cite{FangLiu20} revealed a fundamental principle of quantum mechanics that there exists universal limitations on the accuracy and efficiency of purifying noisy states in general quantum resource theories, by employing one-shot resource theory ideas \cite{LBT19}.
However, Ref.~\cite{FangLiu20} is only part of the story and there are two gaps that we would like to fill to make the picture more complete.  First, the results there assume the input states to be full-rank and it is not fully understood whether there are no-purification rules when the input state is noisy but not of full rank.  Second, the approach developed there is primarily designed for state or static resources, but given that the manipulation of channel or dynamical resources plays intrinsic roles in many scenarios including quantum computation, communication, and error correction, it is also important to understand whether the no-purification principles extend to quantum channels.

In this work, we develop a novel approach to establishing fundamental limits of general quantum resource purification tasks, which addresses the above problems.  This approach is built upon decompositions of the input that separate out the free parts.  As we demonstrate, such  decompositions link the weight of the free parts, a key quantity that we call \emph{free component}, to the optimal error of purification. We apply this approach to both quantum states and channel resource theories.    For state theories, we use the new method to derive new bounds on the error and efficiency of deterministic purification or distillation tasks, which significantly improve those in Ref.~\cite{FangLiu20}.   More specifically, the new results lift the full-rank assumption and imply no-purification principles for a broader range of mixed states. Furthermore, they are quantitatively better and are shown to be tight in certain simple cases. 
We use several concrete examples to demonstrate the improvements and show that the new bounds are tight in certain cases.
Next, as a major contribution of this work, we develop a comprehensive no-purification theory for quantum channels (Ref.~\cite{FangLiu20} presents only a zero-error result).  Most importantly, there are two key complications of the channel theory that does not come up in the state theory: (i) 
There are several different ways to define channel fidelity measures; (ii) Multiple instances of channels can be used or consumed in various presumably inequivalent  ways, such as in parallel, sequentially, or adaptively.   Using the free component method, we derive bounds on the purification errors and costs for all cases.    
To provide a more concrete understanding, we shall discuss the roles and features of common noise channels in different types of channel resource theories, as well as providing guidelines for applying the no-purification bounds to a broad range of fields of great theoretical and practical interest, including distillation, quantum error correction, Shannon theory, and circuit (gate) synthesis.  

\new{We emphasize a particularly remarkable and counterintuitive feature of the no-purification principles, which is that they rule out any noisy-to-pure transformation for noisy input states or channels with free component, where the noisy inputs can be much more ``resourceful''  in terms of common resource measures or operational tasks than the pure targets.  This is in sharp contrast with generic (such as pure-to-pure) transformation tasks where the transformability is naturally determined by the resource content in general.  Also notably, our theory is applicable to virtually all well-defined resource theories (not even requiring the standard convexity assumption), highlighting the fundamental nature of the no-purification principles. }

The paper is organized as follows.  In Sec.~\ref{sec:state}, we apply the free component method to state theories, and in particular discuss the improvements over previous results in Ref.~\cite{FangLiu20}.  In  Sec.~\ref{sec:channel}, we establish the no-purification theory for quantum channels using the free component method.  We first present general-form results in Sec.~\ref{sec:channel general theory}, and then elaborate on specific scenarios and applications in
Sec.~\ref{sec:channel applications}.   Finally in Sec.~\ref{sec:conclusion} we summarize the work and discuss future directions.

\section{State theory}\label{sec:state}

We first consider state resource theories, which are built upon the notions of \emph{free states}  and \emph{free operations} that represent the allowed transformation among states. \new{Here, we consider the most general resource theory framework with the ``minimalist'' requirement---the  \emph{golden rule} that any free operation must map a free state to another free state, or in other words, cannot create resource (see, e.g.,~Refs.~\cite{ChitambarGour19,PhysRevLett.115.070503,PhysRevLett.118.060502}). This golden rule defines the largest possible set of operations that encompasses any legitimate set of free operations, and thus the fundamental limits induced by it apply universally to any nontrivial resource theory.  Also, for mathematical rigor, we assume that  the set of free states $\cF$ has the following two reasonable, commonly held properties:  (i) The composition of free states should be free, namely if $\rho_1, \rho_2 \in \cF$ then  $\rho_1\ox \rho_2 \in \cF$; (ii) $\cF$ is closed.} 

The following quantity that we call \emph{free component} will play a central role in our theory:
\begin{definition}[Free component]
The \emph{free component} of quantum state $\rho$ is defined as
\begin{equation}
\lfc_\rho := \max \left\{\gamma:
\,\rho-\gamma\sigma\geq 0, \sigma\in \cF\right\}.
\end{equation}
\end{definition}
Equivalently,
\begin{equation}
    \lfc_\rho = \max\left\{\gamma: \rho = \gamma\sigma + (1-\gamma)\tau, \sigma\in \cF,\tau \in \mathscr{D} \right\},
\end{equation}
where $\mathscr{D}$ is the set of all density matrices.
That is, the free component is directly related to the ``weight of resource'' $W$, which is recently studied in general resource theory contexts  \cite{Ducuara:weight,Uola:weight}, by $\lfc_\rho = 1-W_\rho$. Another equivalent form is $\lfc(\rho) = \min_{\sigma\in\cF}2^{D_{\max}(\sigma\|\rho)}$ where \new{the max-relative entropy is defined by $D_{\max}(\sigma\|\rho):= \log \min\{t: \sigma \leq t \rho\}$ if $\supp(\sigma) \subseteq \supp(\rho)$ and $+\infty$ otherwise~\cite{datta2009min}.} 
Note that, if $\cF$ can be characterized by semidefinite conditions (which is quite common, e.g., in coherence theory $\cF = \{\sigma: \sigma \geq 0, \tr \sigma = 1, \sigma = \Delta(\sigma)\}$, where $\Delta$ is the dephasing channel erasing the off-diagonal entries), then $\lfc_\rho$ can be efficiently computed by semidefinite programming (SDP) for given $\rho$. {In the resource theory of thermodynamics, \new{for Hamiltonian $H$ and inverse temperature $\beta$ the Gibbs (thermal) state ${\sigma}:= e^{-\beta H}/\tr{e^{-\beta H}}$ is the only free state} and we thus have a closed-form formula for free component as $\lfc_{\rho} = \frac{1}{\lambda^{\max}(\rho^{-1}\new{{\sigma}})}$ (where $\lambda^{\max}$ denotes the largest eigenvalue) if $\supp(\rho) \supseteq \supp(\new{{\sigma}})$ and zero otherwise~\cite[Theorem 2]{rudolph2004quantum}.}

It can be easily seen that the free component obeys the desirable monotonicity property that it cannot be reduced by free operations.
\begin{proposition}[Monotonicity]\label{lem: lfc monotone state}
For any state $\rho$ and any free operation $\cN$, it holds that $\lfc_\rho \leq \lfc_{\cN(\rho)}$.
\end{proposition}
\begin{proof}
Suppose $\lfc_\rho$ is achieved by $\sigma \in \cF$. Then by definition $\rho - \lfc_\rho \sigma \geq 0$, and thus $\cN(\rho) - \lfc_\rho\cN(\sigma) \geq 0$. Since $\cN(\sigma) \in \cF$ by the golden rule, we have that $\lfc_{\cN(\rho)} \geq \lfc_\rho$ by definition.
\end{proof}
Moreover, it is super-multiplicative under tensor product of states:
\begin{proposition}[Super-multiplicity]\label{prop: lfc state super multiplicity}
For any quantum states $\rho_1,\rho_2$, it holds that $\lfc_{\rho_1\ox \rho_2} \geq \lfc_{\rho_1}  \lfc_{\rho_2}$.
\end{proposition}
\begin{proof}
Suppose that the maximization in $\lfc_{\rho_1}, \lfc_{\rho_2}$ are, respectively, achieved by $\sigma_1, \sigma_2\in\cF$, that is, $\rho_i \geq \new{\lfc_{\rho_i}} \sigma_i, i=1,2$. It holds that $\rho_1 \ox \rho_2 \geq (\lfc_{\rho_1}\sigma_1) \ox (\lfc_{\rho_2}\sigma_2) = \lfc_{\rho_1}\lfc_{\rho_2}\sigma_1\ox\sigma_2$. Also note that $\sigma_1 \ox \sigma_2\in\cF$ axiomatically.  Therefore,  $\lfc_{\rho_1\ox \rho_2} \geq \lfc_{\rho_1}  \lfc_{\rho_2}$ by definition.  
\end{proof}






Consider the task of purification, namely transforming a noisy state $\rho$ to a certain target pure state $\psi= \ketbra{\psi}{\psi}$ up to some error.  Formally, the error of purification is defined by the infidelity with the target state: for input state $\rho$, transformation operation $\cN$ and target pure state $\psi$, the error 
\begin{align}
\varepsilon(\rho\xrightarrow{\cN}\psi) := 1-\tr\psi\cN(\rho) = 1-\bra{\psi}\cN(\rho)\ket{\psi}.
\end{align}
Also, let $f_{\psi}$ denote the maximum overlap of pure state $\psi = \ketbra{\psi}{\psi}$ with free states, namely, \begin{align}
f_{\psi}:= \max_{\sigma \in \cF} \tr \psi \sigma = \max_{\sigma \in \cF}\bra{\psi}\sigma\ket{\psi}.   
\end{align}
We now prove an improved deterministic no-purification theorem using a method different from Ref.~\cite{FangLiu20}, which directly connects the accuracy of purifying a noisy state with its free component.

\begin{theorem}\label{thm: nogo deterministic state transformation}
Given any state $\rho$ and any pure state $\psi$, there is no free operation that transforms $\rho$ to $\psi$ with error smaller than $\lfc_\rho(1-f_{\psi})$.  That is, it holds for any free operation $\cN$ that
\begin{align}
    \varepsilon(\rho\xrightarrow{\cN}\psi) \geq \lfc_\rho(1-f_\psi).
\end{align}
\end{theorem}

\begin{proof}
By the definition of $\lfc_\rho$, there exists free state $\sigma\in\cF$ and state $\tau$ such that $\rho$ can be decomposed as follows: 
\begin{equation}
    \rho = \lfc_\rho\sigma + (1-\lfc_\rho)\tau.
\end{equation}
Let $\cN$ be any free operation.  By linearity,
\begin{equation}
    \cN(\rho) = \lfc_\rho \cN(\sigma) + (1-\lfc_\rho)\cN(\tau).
\end{equation}
Then it holds that
\begin{align}
\varepsilon(\rho\xrightarrow{\cN}\psi)&= 1-\tr\cN(\rho)\psi \\&= 1- \lfc_\rho\tr\cN(\sigma)\psi - (1-\lfc_\rho)\tr\cN(\tau)\psi \\ &\geq 1- \lfc_\rho f_{\psi} - (1-\lfc_\rho)\\
&=\lfc_\rho(1-f_\psi),
\end{align}
where the inequality follows from  $\tr\cN(\sigma)\psi \leq f_\psi$ since $\cN(\sigma)\in\cF$ by the golden rule, and $\tr\cN(\tau)\psi \leq 1$.
\end{proof}
As first noted in Ref.~\cite{FangLiu20}, we can translate the upper bounds on transformation accuracy into lower bounds on the ``amount'' of input resources required to achieve a certain target, in particular, the cost of many-copy distillation procedures, which are widely considered for various purposes in quantum computation and information \cite{BBPS96:ent_dist,BBPSSW96:ent_dist,BDSW96:ent_qec,Bravyi2005}.
The above Theorem~\ref{thm: nogo deterministic state transformation} induces the following general lower bound on distillation overhead.
\begin{corollary}
Consider distillation procedures represented by a free operation that transform $n$ copies of noisy states $\rho$ to a target pure state $\psi$ within error $\epsilon$. Then $n$ must satisfy:
\begin{equation}\label{eq: state overhead}
    n\geq \Bigg[\log\frac{1-f_\psi}{\epsilon}\Bigg]\Bigg[\log\frac{1}{\lfc_{\rho}}\Bigg]^{-1}.
\end{equation}
\label{cor:overhead}
\end{corollary}

\begin{proof}
Suppose the transformation is given by the free operation $\cN$. Then it holds from Theorem~\ref{thm: nogo deterministic state transformation} that
\begin{align}
\epsilon \geq \ve(\rho^{\ox n}\xrightarrow{\cN}\psi) \geq \lfc_{\rho^{\ox n}}(1-f_\psi).
\end{align}
Note that,  due to  the super-multiplicity property from Proposition~\ref{prop: lfc state super multiplicity}, we get $\lfc_{\rho^{\ox n}}\geq (\lfc_{\rho})^n$. This gives 
\begin{align}
\epsilon \geq (\lfc_{\rho})^n (1-f_\psi),
\end{align}
which is equivalent to the above assertion.
\end{proof}

Our new method essentially replaces the minimum eigenvalue of $\rho$ in the corresponding bounds in Ref.~\cite{FangLiu20} (which we refer to as the min-eigenvalue bounds) by its free component, which represents a significant improvement from both qualitative and quantitative perspectives, as detailed in the following.

First, the range of applicability of the no-purification theorem is significantly extended.  The proof using the quantum hypothesis testing relative entropy presented in Ref.~\cite{FangLiu20} applies only to full-rank input states.  However,  Theorem~\ref{thm: nogo deterministic state transformation} implies that the no-purification rule actually holds more broadly (see also Ref.~\cite[Proposition 2]{RBTL20:benchmarking}):
\begin{corollary}
There is no free operation that exactly transforms a state $\rho$ to any pure state $\psi\notin\cF$  if $\lfc_\rho > 0$.
\end{corollary}
\begin{proof}
Since $\cF$ is closed by assumption and $\psi \notin \cF$, we have $f_\psi < 1$. Then due to Theorem~\ref{thm: nogo deterministic state transformation}, the transformation error $\ve > 0$, indicating that exact transformation is impossible.
\end{proof}

Below we give some alternative useful characterizations of the $\lfc > 0$ condition.
\begin{proposition}\label{prop:no-go conditions}
For any quantum state $\rho$, the following conditions are equivalent:
\begin{enumerate}
    \item[(a)] Free component: $\lfc_\rho>0$;
    \item[(b)] Support: There exists a free state $\sigma \in \cF$ such that the support condition $\mathrm{supp}(\rho) \supseteq \mathrm{supp}(\sigma)$ holds;
    \item[(c)] Resource measure: The min-relative entropy of resource $\mathfrak{D}_{\min}(\rho):= \min_{\sigma \in \cF} D_{\min}(\rho\|\sigma) = 0$, where $D_{\min}(\rho\|\sigma) := -\log\tr\Pi_\rho\sigma$ is the min-relative entropy, and $\Pi_\rho$ is the projector onto $\supp(\rho)$.
\end{enumerate}
\end{proposition}
\begin{proof}	
It is easy to see that (a) implies (b) and (b) implies (c). It remains to show that (c) implies (a). Suppose (c) holds, then there exists $\sigma \in \cF$ such that $D_{\min}(\rho\|\sigma) = 0$. By definition, we have  $\tr (I - \Pi_{\rho}) \sigma = 0$. That is, $\ker(\rho) \perp \mathrm{supp}(\sigma)$, and equivalently, $\mathrm{supp}(\rho) \supseteq \mathrm{supp}(\sigma)$. If $\rho = \sigma$, then (a) holds. Otherwise, we have $D_{\max}(\sigma\|\rho) > 0$. (Note that $D_{\max}(\sigma\|\rho) = 0$ if and only if $\rho = \sigma$.) This implies that there exists $t > 1$ such that $\sigma \leq t \rho$. \new{Thus $\Gamma_\rho \geq 1/t > 0$}, implying (a).
\end{proof}
It is  clear  that for any pure resource state $\psi\notin\cF$ we have  $\lfc_\psi = 0$, so the no-purification bounds can only be nontrivial for mixed states.
Meanwhile, it can be immediately seen (e.g.~from (b)) that the $\lfc > 0$ condition is weaker than the full-rank condition.  In fact, it holds as long as the support of $\rho$ contains some free state in its support, which is generically the case for mixed states in common resource theories.
Also note that the $\lfc > 0$ condition does not necessarily hold for all mixed states.  For a concrete example, consider the coherence theory defined by an orthonormal basis $\{\ket{0},\ket{1},\ket{2},\ket{3}\}$. Consider the state $\rho = (\ket{\psi_1}\bra{\psi_1}+\ket{\psi_2}\bra{\psi_2})/2$ where $\ket{\psi_1}=(\ket{0}+\ket{1})/\sqrt{2}, \ket{\psi_2}=(\ket{2}+\ket{3})/\sqrt{2}$. It can be verified that $\rho$ is mixed but $\lfc_\rho = 0$ because any decreasing of the diagonal entries will render the matrix negative. It would be interesting to further understand and characterize the $\lfc > 0$ condition in specific theories.

\new{Furthermore, note that the derivation and results (also the channel versions below) apply to continuous variable or infinite-dimensional quantum systems: the relevant quantities, the free component $\Gamma$ and the maximum overlap $f$, can be defined likewise (supremum instead of maximum over $\cF$), and the proof steps follow.  In particular, $\Gamma > 0, f<1$ still indicate no-purification.  An elementary continuous variable example  will be given later.  }

We remark that if we only require the purification transformation to succeed with some probability (the probabilistic setting), the $\lfc>0$ condition is not sufficient to rule out purification and it seems that the full-rank condition cannot be alleviated.  For example, consider the following state with a flag register $F$:
\begin{equation}
    \rho = p|0\rangle\langle 0|_F \otimes \psi_{A} + (1-p)|1\rangle\langle 1|_F \otimes \tau_{A},
\end{equation}
where $\psi$ is the target pure state and $\tau$ is a state such that $\lfc_\tau > 0$.   Then we have $\lfc_\rho>0$  ($\rho$ is not full-rank), but we can obtain $\psi$ with probability $p$ simply by measuring $F$ (which is conventionally free) and postselect on 0. 

Second, the new results are quantitatively better than the corresponding ones in Ref.~\cite{FangLiu20}  for full-rank input states.  It is first straightforward to see that $\lfc_\rho \geq  \lambda^{\min}_\rho$, where $\lambda^{\min}_\rho$ denotes the minimum nonzero eigenvalue of $\rho$, because
$\rho \geq \lambda_\rho^{\min} \cdot I \geq \lambda_\rho^{\min} \cdot \sigma$ for any state $\sigma$ where $I$ denotes the identity matrix on $\supp(\rho)$. So by definition, $\lfc_\rho \geq \lambda^{\min}_\rho$.  In sum, the new free component bounds  cover the min-eigenvalue bounds.
In particular, when the noisy state $\rho$ is close to the set of free states $\cF$,  the minimum eigenvalue  $\lambda^{\min}_\rho$ could still be small but $\lfc_\rho$ approaches one. This indicates that the free component bounds potentially exhibit much tighter behaviors in the large error regime like when $\rho$ is close to $\cF$. 
Importantly, the distillation overhead bound Corollary~\ref{cor:overhead} indicates the key behavior that as $\rho$ approaches $\cF$, it holds that $n\rightarrow \infty$, i.e.~the number of copies needed diverges,  because $\lfc_{\rho}\rightarrow 1$. This cannot be deduced from the min-eigenvalue bounds in Ref.~\cite{FangLiu20}.   

Now we discuss the application of our general bounds in a few important specific scenarios that are of practical interest in diverse manners, showcasing the versatility of our theory.  In particular, it is concretely demonstrated that the free component bounds can strictly outperform the corresponding min-eigenvalue bounds in Ref.~\cite{FangLiu20} and notably, could be tight,  in key scenarios.
\begin{exmp}[Magic state distillation]
Consider $T$ states $\ket{T} = T\ket{+} = (\ket{0}+e^{i\pi/4}\ket{1})/\sqrt{2}$ contaminated by depolarizing or dephasing noise,
given by
\begin{equation}\label{eq: noisy magic state}
    \tau = (1-\zeta)\ket{T}\bra{T} + \zeta\frac{I}{2}, 
\end{equation}
where $\zeta$ is the noise rate, as the input.  Note that we are interested in $\zeta\in(0,1-1/\sqrt{2})$ so that $\tau$ is a mixed state outside of the stabilizer hull. On the one hand, it can be directly checked that $\lambda^{\min}_\tau = \zeta/2$.  On the other hand, $\lfc_\tau$ is bounded as follows. Consider the free state
\begin{equation}\notag
    \bar\tau =  \frac{1}{2}(\ket{+}\bra{+} + T^2\ket{+}\bra{+} {T^\dagger}^2) = \frac{1}{4}\begin{pmatrix}
2 & {1-i}\\
{1+i} & 2
\end{pmatrix} \in \cF,
\end{equation}
which sits at the edge of the stabilizer hull closest to $\ket{T}$ (as depicted in Fig.~\ref{fig:bloch}).
\begin{figure}
\centering
\includegraphics{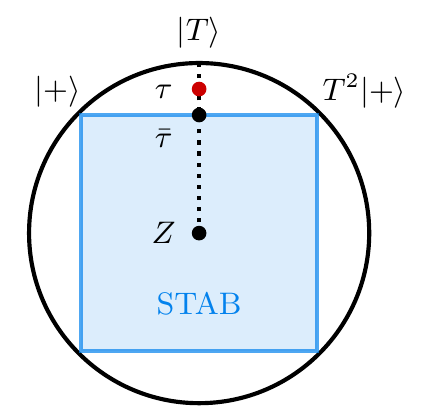}
\caption{The cross section of the Bloch sphere through the center perpendicular to the $Z$ axis. The blue square represents the corresponding cross section of the stabilizer hull (STAB). $\bar\tau$ is actually $\ket{T}$ subject to $p=1-1/{\sqrt{2}}$ depolarizing noise and lies on the edge of STAB. $\tau$ is the noisy input state that lies between $\ket{T}$ and $\bar\tau$.}
 \label{fig:bloch}
 \end{figure}
 Then by definition we have 
\begin{align}\notag
\lfc_\tau \geq \max\{\gamma: \tau - \gamma \bar\tau \geq 0\} = \max\left\{\gamma: \begin{pmatrix}
\alpha & \beta\\
\bar \beta & \alpha
\end{pmatrix} \geq 0\right\},
\end{align}
with $\alpha = \frac{1}{2}(1-\gamma)$ and $\beta = \frac{1-\zeta}{2}e^{-i\pi/4}-\gamma\frac{1-i}{4}$.
By solving the determinant we obtain that
\begin{equation}\label{eq: tau lower bound}
\lfc_\tau \geq (2+\sqrt{2})\zeta,
\end{equation}
when $\zeta\in(0,1-1/\sqrt{2})$.
This implies $\lfc_\tau > \lambda^{\min}_\tau$, and thus the previous error bound is outperformed for any pure target state by a constant factor.  As a sanity check, the bound indeed approaches 1 as $\zeta \rightarrow 1-1/\sqrt{2}$, in contrast to the $\lambda^{\min}$ bound.
This indeed implies the expected phenomenon that the total distillation overhead blows up as $\tau$ approaches the stabilizer hull. 
In particular, for the standard task of distilling $T$ states, we thus obtain an improved bound on the average overhead following the proof of Theorem 3 in Ref.~\cite{FangLiu20}. 
\begin{corollary}
Consider the following common formulation of magic state distillation task: given $n$ copies of noisy states $\tau$ (defined in Eq.~\eqref{eq: noisy magic state}), output an $m$-qubit state $\sigma$ such that $\tr\sigma_i T = \bra{T}\sigma_i\ket{T}\geq 1-\epsilon, \forall i = 1,\cdots, m$ where $\sigma_i = \tr_{\bar{i}}\sigma$ is the $i$-th qubit, by some free (stabilizer-preserving) operation. Then $n$ must satisfy:
\begin{align}
n &\geq \Bigg[\log\frac{(4-2\sqrt{2})^m-1}{(4-2\sqrt{2})^m m\epsilon}\Bigg]\Bigg[\log\frac{2-\sqrt{2}}{2\zeta}\Bigg]^{-1}.
\end{align}
\end{corollary}
\begin{proof}
By applying the union bound, we have $\bra{T^{\otimes m}}\sigma\ket{T^{\otimes m}} \geq 1-m\epsilon$. Recall that $\lfc_\tau \geq (2+\sqrt{2})\zeta$ in Eq.~\eqref{eq: tau lower bound} and notice that $f_{T^{\otimes m}} = (4-2\sqrt{2})^{-m}$ \cite{LBT19,Campbell11:catalysis,BravyiGosset16,bravyi2019simulation}. By plugging everything into Eq.~({\ref{eq: state overhead}}) we obtain the claimed bound.
\end{proof}
\end{exmp}

\begin{exmp}[Coherence]
Consider the maximally coherent qubit state $\ket{+}=(\ket{0}+\ket{1})/\sqrt{2}$ contaminated by typical noise channels, including depolarizing, dephasing, and amplitude damping, as the input. 

For the depolarizing noise (here the dephasing noise has an equivalent effect), the noisy state is given by $\rho = [\ketbra{0}{0}+\ketbra{1}{1}+(1-\mu)(\ketbra{0}{1}+\ketbra{1}{0})]/2$ where $\mu\in(0,1)$ is the noise rate.  Then $\lambda^{\min}_\rho = \mu/2$, and it can be easily calculated that $ \lfc_\rho = \mu$. 
That is, the new error bound is twice the min-eigenvalue bound for any pure target state. 

For the amplitude damping noise, the free component bounds have a more remarkable advantage.  Here the noisy state is given by 
$\rho = [(1+\nu)\ketbra{0}{0} + (1-\nu)\ketbra{1}{1} + \sqrt{1-\nu}(\ketbra{0}{1}+\ketbra{1}{0})]/2$
where $\nu\in(0,1)$ is the noise rate. Then $\lambda^{\min}_\rho = \frac{1}{2}(1-\sqrt{1-\nu+\nu^2})$. 
We numerically solve $\lfc_\rho$, and compare it with $\lambda^{\min}_\rho$  in Fig.~\ref{fig:noise}(b) (note that the values plotted are all multiplied by a factor $1/2$; see below). 
Note that, as $\nu$ increases, i.e.~$\rho$ is more heavily damped towards the free state $\ket{0}$, the error of purification is expected to grow. As can be seen from Fig.~\ref{fig:noise}(b), as $\nu\rightarrow 1$, $\lambda^{\min}_\rho$ vanishes and so do corresponding bounds, but $\lfc_\rho$ indeed keeps growing, showcasing an important scenario where only the free component bounds are nontrivial.

Let us explicitly consider $\ket{+}$ as the target state. 
It is known that the optimal fidelity of transforming $\rho$ to $\ket{+}$ by free operations (MIO) can be solved by the following SDP~\cite[Theorem 3]{regula2018one}:
\begin{align}
    \max \left\{\tr G\rho : 0 \leq G \leq I, \Delta(G) = \frac{I}{2} \right\},
    \label{eq:mio-sdp}
\end{align}
where $\Delta$ takes the diagonal part of a given matrix.
In Fig.~\ref{fig:noise}, we plot the optimal error obtained by the above SDP as well as the free component and min-eigenvalue lower bounds for comparison.  
In particular, for depolarizing or dephasing noise, the free component error bound turns out to be tight, i.e.~is achievable, for any noise rate.  

\begin{figure}[]
\centering
\includegraphics[]{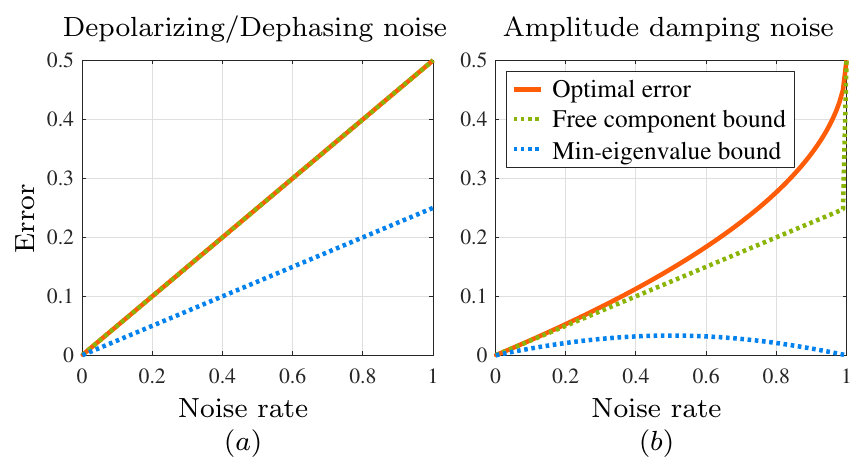}

\caption{Comparisons between the optimal achievable error of a standard purification task and the lower bounds induced by $\lfc$ (this work) and $\lambda^{\min}$ (Ref.~\cite{FangLiu20}) in coherence theory. The task is to recover the maximally coherent qubit state $\ket{+}$ under typical noise channels: (a) depolarizing and dephasing; (b) amplitude damping. The green and blue dashed lines are respectively the free component and min-eigenvalue lower bounds on the error, and the red line is the minimum error achieved by MIO computed by SDP Eq.~(\ref{eq:mio-sdp}).  In (a) the green dashed line actually overlaps with the red line, indicating that the free component error bound is tight.   
}\label{fig:noise}
\end{figure}

\end{exmp}
\begin{exmp}[Constrained quantum error correction]
Here we demonstrate how  the state no-purification bounds can be used to find limits on quantum error correction (QEC).  In particular, we consider the broadly important situations where the QEC procedures are subject to certain constraints (such as stabilizer or Clifford constraints, symmetries) so that resource theory becomes useful.  Notice that the  decoding procedures are aimed at recovering all logical states from noisy physical states, indicating connections between the no-purification bounds and the overall recovery accuracy.  More specifically, we have the general result ($L, S$ denote the logical and physical systems, respectively):
\begin{corollary}
  Suppose the decoding operation is free.  Then given encoding operation $\cE_{L\to S}$ and noise channel $\cN_S$ acting on the physical system $S$, the error of the recovery of pure logical state $\psi_L$ obeys  $\varepsilon\geq{\Gamma_{\cN_S \circ \cE_{L\to S}({\psi_L})}}(1-f_{\psi_L})$, based on which we directly obtain bounds on measures of the overall accuracy of the code, such as the worst-case error given by maximization over $\psi_L$, and the average-case  error given by a certain (e.g.~Haar) average over $\psi_L$.
\label{cor:qec state}
\end{corollary}
We further remark on the case of covariant (symmetry-constrained) codes, which play fundamental roles in quantum computing and physics and has drawn considerable recent interest \cite{hayden2017error,faist2019continuous,woods2019continuous,kubica2020using,ZhouLiuJiang20,YangMo20:covariant,KongLiu}.  Suppose we consider some compact continuous symmetry group $G$.  Based on Lemma 2 in  Ref.~\cite{ZhouLiuJiang20}\footnote{Replace $\frac{1}{\tau}\int_0^\tau d\theta$ by the integration over the Haar measure on $G$.}, it can be seen that when the noise channel $\cN_S$ is covariant (which is usually the case), then we can construct a covariant decoding operation that achieves the optimal error. That is, we can actually remove the freeness assumption of the decoder to apply the no-purification bounds, leading to the following adapted version: 
\begin{corollary}[Covariant code]
Let $G$ be a compact continuous symmetry group.   Let $\cE_{L\to S}$ be a $G$-covariant encoding operation.  Suppose the noise channel $\cN_S$ is $G$-covariant. Then Corollary~\ref{cor:qec state} (where the parameters are defined in terms of the $G$-asymmetry theory) holds for any decoder.
\label{cor:covariant state}
\end{corollary}
See Sec.~\ref{sec:qec} for related discussions and results in the channel setting.
\end{exmp}

\new{
\begin{exmp}[Continuous variable]
Lastly, we provide an elementary example of the application to continuous-variable theories.  Consider  continuous-variable nonclassicality, a characteristic resource feature  in quantum optics that is closely relevant to, e.g.,~linear optical quantum computation \cite{Knill2001} and metrology \cite{SCHNABEL17:squeezed,Yadin18:nonclassicality,PhysRevLett.122.040503}.  Here the coherent states of light 
and their probabilistic mixtures are considered free (classical). The coherent state corresponding to complex amplitude $\alpha\in\mathbb{C}$ takes the form 
\begin{equation}
    |\alpha\rangle=\exp \left(-|\alpha|^{2}/2\right) \sum_{n=0}^{\infty} \frac{\alpha^{n}}{\sqrt{n !}}|n\rangle
\end{equation}
in the number state (Fock) basis $\{\ket{n}\}$.
A prototypical type of nonclassical resource states is the (single-mode) squeezed states \cite{PhysRevA.13.2226,rmp12:gaussian}
\begin{equation}
    |s_r\rangle:=\frac{1}{\sqrt{\cosh r}} \sum_{n=0}^{\infty} \frac{\sqrt{(2 n) !}}{2^{n} n !} (\tanh r)^n|2 n\rangle
\end{equation}
generated by the squeezing operator
$S(r):=\exp \left[r\left(\hat{a}^{2}-(\hat{a}^{\dagger})^2\right) / 2\right]$ ($\hat{a}$ and $\hat{a}^{\dagger}$ are, respectively, the annihilation and creation operators)
acting on the vacuum state $\ket{0}$, where  $r\geq 0$ is the squeezing parameter. It can be calculated that
\begin{align}
    \langle s_r|\alpha\rangle &=  \sqrt{\frac{e^{-|\alpha|^{2}}}{{\cosh r}}}\sum_{n=0}^{\infty} \frac{\sqrt{(2 n) !}}{2^{n} n !} (\tanh r)^n\cdot\frac{\alpha^{2n}}{\sqrt{(2n)!}}\\
    &= \sqrt{\frac{e^{-|\alpha|^{2}}}{{\cosh r}}}\sum_{n=0}^{\infty} \frac{[(\alpha^2 \tanh r)/2]^{n} }{n !}\\
    &=\sqrt{\frac{e^{-|\alpha|^{2}}}{{\cosh r}}}e^{(\alpha^2 \tanh r)/2},
\end{align}
using which we obtain
\begin{align}
    f_{s_r} &:= \sup_{\sigma \in \cF}  \bra{s_r}\sigma\ket{s_r} = \sup_{\alpha \in \mathbb{C}}|\langle s_r|\alpha\rangle|^2 \\
    & = \sup_{\alpha \in \mathbb{C}} \frac{e^{-|\alpha|^{2}}}{{\cosh r}} e^{\mathrm{Re}(\alpha^2) {\tanh r}} \\
    & = (\cosh r)^{-1},
\end{align}
where we used $\tanh r < 1$.
Then, to showcase an example of a no-purification bound,  consider the task of distilling some squeezed state $|s_r\rangle$ from noisy state $\rho$ using free, namely classicality-preserving operations (which, in particular, include passive linear optical operations) \cite{Yadin18:nonclassicality,lami18:gaussianrt}. Then  Theorem~\ref{thm: nogo deterministic state transformation} directly implies that the transformation error $\varepsilon \geq \lfc_\rho[1- (\cosh r)^{-1}]$, from which it can be observed that the task indeed becomes more demanding as the squeezing parameter increases.  Like  the discrete-variable setting, for specific noise models, it is often easy to calculate or bound $\lfc_\rho$ so that the error bound can be further specified. 
 
\end{exmp}
}


\section{Channel theory}\label{sec:channel}


We now extend the no-purification theory to quantum channels or dynamical settings. 
The channel analog of purification is to transform a noisy channel (or noisy channels, as will be discussed) to a unitary (noiseless) channel, or equivalently, to simulate the unitary channel by the noisy ones.  
The free component approach directly enables us to study these problems in the channel resource theory setting where the resource objects are quantum channels instead of states (note that it is not clear how to fully extend the hypothesis testing approach in Ref.~\cite{FangLiu20} to channels). It is worth noting again that the structure of channel theories is much richer than the state theories since multiple instances of channels can be used in different ways, such as in parallel, sequentially, or adaptively.
Here, we first present error bounds in the most general forms, and then specifically investigate the adaptive or sequential simulation setting, which represents a fundamental difference from state theories.  
To demonstrate the practical relevance of the general no-go rules and bounds, we discuss them in more specific contexts, and, in particular, outline the applications to quantum error  correction, gate and circuit synthesis, and channel capacities.   

Note that we often specify the input and output systems of channels in the subscripts (a channel $\cN$ from system $A$ to system $B$ is denoted as $\cN_{A\rightarrow B}$, and if the input and output systems are the same one $A$ it is simply denoted as $\cN_A$), but when there is no ambiguity we shall omit the labels. Given linear maps $\cN,\cM$, the order $\cN - \cM \geq 0$ means $\cN - \cM$ is a completely positive map.  To simplify the notation,  given some input state $\rho$ on $A$ and reference system $R$, we will also denote the output state of the channel $\cN_{A\rightarrow B}$ acting on $A$ by 
\begin{align}\label{eq:outputstate}
\rho_\cN := \cN_{A\to B}\otimes\id_R(\rho_{AR}).
\end{align}
In particular, the Choi state of $\cN$ is given by 
\begin{align}\label{eq:choistate}
\Phi_\cN:=\cN_{A\to B}\otimes\id_R(\Phi_{AR}),
\end{align}
where $\Phi_{AR} = \sum_j\ket{j}_A\ket{j}_R/\sqrt{d}$ is the maximally entangled state between $A$ and reference system $R$ of the same dimension $d$.

\subsection{General theory and results}\label{sec:channel general theory}
\subsubsection{Setups and basic error bounds}

For channel resource theories, the building blocks analogous to free states and free operations are \emph{free channels} and \emph{free superchannels}, where superchannels map channels to channels. 
\new{Like the state case, we consider the most general framework where the free superchannels are required only to obey the  \emph{golden rule} that any free superchannel must map a free channel to another free channel. Note again that this golden rule gives rise to the largest possible set of superchannels that encompasses any legitimate set of free superchannels, so that the fundamental limits induced by it apply universally.  We also assume the following two commonly held properties of the set of free channels (which we still denote by $\cF$):  (i) The composition of free channels (for channels there are two fundamental types of composition---parallel composition (represented by tensor product $\ox$), and sequential composition (represented by $\circ$)) should be free, that is, if $\cN_1, \cN_2 \in \cF$, then both $\cN_1\ox \cN_2 \in \cF$ and $\cN_2 \circ \cN_1\in \cF$ hold; (ii) $\cF$ is closed.} 
We refer readers to e.g.~Refs.~\cite{LiuWinter19,gour2019entanglement} for more comprehensive discussions of the general framework of channel resource theories.

We now define the channel version of free component as follows.
\begin{definition}[Channel free component]
The \emph{free component} of quantum channel $\cN$ is defined as
\begin{align}
    \lfc_\cN:= \max \big\{\gamma : \cN - \gamma \cM \geq 0, \ \cM\in\cF\big\}.
\end{align}
\end{definition}
Equivalently, 
\begin{align}
\lfc_\cN = \max \big\{\gamma : \cN = \gamma \cM + & (1-\gamma)\cR,\; \cM\in\cF,\cR \in \mathscr{C}\big\},
\end{align}
where $\mathscr{C}$ is the set of all completely positive and trace-preserving maps (quantum channels). Since $\cN \geq \gamma\cM$ is  equivalent to $\Phi_\cN  \geq \gamma\Phi_\cM$, we also have the relation
\begin{align}
\lfc_\cN = \lfc_{\Phi_\cN},
\end{align}
where on the RHS, $\Phi_\cN$ is the Choi state of $\cN$ and the free component $\Gamma$ is defined with respect to the set of free states consisting of the Choi states of all free channels.
Similar to the state case, 
as long as $\cF$ can be characterized by semidefinite conditions,  the channel free component $\lfc_{\cN}$ can be efficiently computed by SDP.

The channel free component also exhibits  monotonicity and super-multiplicity properties.

\begin{proposition}[Monotonicity]\label{lem: lfc channel monotonicity}
For any quantum channel $\cN$ and free superchannel $\Pi$, it holds that 
\begin{align}
\lfc_\cN \leq \lfc_{\Pi(\cN)}. 
\end{align}
\end{proposition}
\begin{proof}
Suppose $\lfc_\cN$ is achieved by $\cM \in \cF$. Then by definition $\cN - \lfc_\cN \cM \geq 0$, and thus $\Pi(\cN) - \lfc_\cN\Pi(\cM) \geq 0$. Since $\Pi(\cM) \in \cF$ by the golden rule, we have that $\lfc_{\Pi(\cN)} \geq \lfc_\cN$ by definition.
\end{proof}

For channels, we need to consider sequential composition in addition to parallel composition represented by tensor product. The channel free component is super-multiplicative under both types of composition.
\begin{proposition}[Super-multiplicity]\label{prop: lfc channel super multiplicity}
For any quantum channels $\cN_1,\cN_2$, it holds that 
\begin{align}
\lfc_{\cN_1\ox \cN_2} & \geq \lfc_{\cN_1}  \lfc_{\cN_2}\\
\lfc_{\cN_2\circ \cN_1} & \geq \lfc_{\cN_1} \lfc_{\cN_2}. 
\end{align}
\end{proposition}
\begin{proof}
Suppose that the maximization in $\lfc_{\cN_1}, \lfc_{\cN_2}$ are, respectively, achieved by $\cM_1, \cM_2\in\cF$, that is, $\cN_i \geq \new{\lfc_{\cN_i}} \cM_i, i=1,2$. It holds that $\cN_1 \ox \cN_2 \geq (\lfc_{\cN_1}\cM_1) \ox (\lfc_{\cN_2}\cM_2) = \lfc_{\cN_1}\lfc_{\cN_2}\cM_1\ox\cM_2$, and similarly, $\cN_2 \circ \cN_1 \geq (\lfc_{\cN_2}\cM_2) \circ (\lfc_{\cN_1}\cM_1) = \lfc_{\cN_1}\lfc_{\cN_2}\cM_2\circ\cM_1$.  Also note that $\cM_1 \ox \cM_2\in\cF$ and $\cM_1 \circ \cM_2\in\cF$ axiomatically.  Therefore,  $\lfc_{\cN_1\ox \cN_2} \geq \lfc_{\cN_1}  \lfc_{\cN_2}$ and $\lfc_{\cN_2\circ \cN_1} \geq \lfc_{\cN_1} \lfc_{\cN_2}$ by definition.  
\end{proof}


Here we are interested in the channel simulation task of transforming a given quantum channel $\cN$ to a target unitary channel $\cU$ via some superchannel up to some error that is measured by certain choices of channel distances. Let $F(\rho,\sigma) = \|\sqrt{\rho}\sqrt{\sigma}\|_1^2$ be the Uhlmann fidelity between general states $\rho$ and $\sigma$. Consider the following three typical versions of channel fidelity that are commonly used.  
\begin{itemize}
\item Worst-case (entanglement) fidelity:
\begin{align}
F_{\wst}(\cN,\cM):=\inf_{\rho_{AR}} F(\rho_\cN, \rho_\cM),
\end{align}
where $\rho_\cN,\rho_\cM$ are, respectively, the channel output states of $\cN,\cM$ defined in Eq.~(\ref{eq:outputstate}),  and the optimization includes system $R$.  Note that it is equivalent to optimize over pure input states due to the joint concavity of fidelity $F$ \cite{wilde_2013}.  
\item Choi fidelity:  
\begin{align}
F_{\cho}(\cN,\cM) :=F(\Phi_\cN, \Phi_\cM)
\end{align}
where $\Phi_\cN, \Phi_\cM$ are, respectively, the Choi states of $\cN,\cM$.
\item Average-case fidelity \cite{gilchrist2005distance}:
\begin{align}
F_{\ave}(\cN,\cM):=\int d \psi~ F(\cN(\psi),\cM(\psi)),
\end{align}
where the  integral is over the Haar measure on the input state space.
\end{itemize}
The corresponding versions of infidelity are then
\begin{align}
    \varepsilon_{x}(\cN, \cM) := 1- F_{x}(\cN,\cM),\quad x \in \{\wst,\cho,\ave\}.
\end{align}
Also, a standard measure of distance between channels is given by the diamond norm distance: 
\begin{align}
\varepsilon_\diamond(\cN, \cM):= \frac{1}{2}\diamnorm{\cN - \cM},
\end{align}
where $\diamnorm{\cN} := \sup_{\rho_{AR}}\|\cN_{A\to B}\otimes\id_R(\rho_{AR})\|_1$. Again, it is equivalent to optimize over pure input states due to the convexity of trace norm $\|\cdot\|_1$.
All the above channel distance measures are symmetric in its arguments.

Note that these channel distance measures are commonly used in different scenarios~\cite{gilchrist2005distance}. For example, the worst-case entanglement fidelity and the diamond norm error are commonly used in quantum computation scenarios like circuit synthesis (see, e.g.,~Refs.~\cite{Kitaev_1997,10.5555/2011679.2011685}, Sec.~\ref{sec:synthesis}) and approximate quantum error correction (see, e.g.,Ref.~\cite{leung1997approximate}, Sec.~\ref{sec:qec}); the Choi fidelity  is used in quantum Shannon theory to evaluate the performance of quantum communication (see, e.g.,~Refs.~\cite{tomamichel2016quantum,wang2018semidefinite}, Sec.~\ref{sec:capacities});  the average-case fidelity is easier to estimate in experiments (see, e.g.,Refs.~\cite{bowdrey2002fidelity,chow2009randomized,flammia2011direct,lu2015experimental}).

In this work, we are mostly interested in the case where an argument is a unitary channel $\cU$.
Note that for pure state $\psi$, we have the inequality  \cite{nielsen2010quantum}
\begin{align}
     1-F(\rho,\psi) \leq \frac{1}{2}\|\rho-\psi\|_1 \leq \sqrt{1-F(\rho,\psi)}.
\end{align}
Applying the above result to channels, we can conclude 
\begin{align}
    \varepsilon_\wst(\cN, \cU) \leq \varepsilon_\diamond(\cN, \cU) \leq \sqrt{\varepsilon_\wst(\cN, \cU)}.  \label{eq:diamond_vs_worst}
\end{align}
Also, it is known~\cite{1998quant.ph..7091H,gilchrist2005distance} that the average-case fidelity and the Choi fidelity have the following direct relation: 
\begin{align}\label{eq: average case Choi state case equivalent}
	F_{\ave}(\cN,\cU) = \frac{F_{\cho}(\cN,\cU) d + 1}{d+1},
\end{align}
and thus
\begin{align}
    \varepsilon_\ave(\cN, \cU)  = \frac{d}{d+1}\varepsilon_\cho(\cN, \cU),
\end{align}
where $d$ is the dimension of the input system. 
Furthermore, it is clear from definition that 
\begin{align}
 F_\cho(\cN, \cM) \geq F_\wst(\cN, \cM),\\
    \varepsilon_\cho(\cN, \cM) \leq \varepsilon_\wst(\cN, \cM),
\end{align}
for any channels $\cN, \cM$.
To summarize, for the case  of comparing with unitary channel $\cU$ which is of interest in this work,  the four channel distance measures are ordered as follows:  
\begin{equation}
    \varepsilon_\diamond(\cN, \cU) \geq \varepsilon_\wst(\cN, \cU) \geq \varepsilon_\cho(\cN, \cU) = \frac{d+1}{d}\varepsilon_\ave(\cN, \cU).
    \label{eq:ordering}
\end{equation}

We are interested in the task of using channel $\cN$ to simulate unitary target channel $\cU$ via transformation superchannel $\Pi$.
The (different versions of) simulation error is simply given by 
\begin{align}
    \varepsilon_{x}(\cN\xrightarrow{\Pi}\cU) := \varepsilon_{x}(\Pi(\cN),\cU),\quad x \in \{\wst,\cho,\ave,\diamond\}.
\end{align}
Also define corresponding versions of the maximum overlap of channel $\cN$ with free channels as
\begin{align}
	f^{x}_\cN := \max_{\cM \in \cF} F_{x}(\cN,\cM),\quad x \in \{\wst,\cho,\ave\}.
\end{align}

Note the following simple fact.
\begin{proposition}[Faithfulness]\label{rmk: fidelity faithfulness}
 For any quantum channels $\cN$ and $\cM$, for $x \in \{\wst,\cho,\ave\}$,
\begin{align}
F_{x}(\cN,\cM)= 1 \iff \cN=\cM,
\end{align}
and as a consequence,
\begin{align}
    f^{x}_{\cN}= 1 \iff \cN \in \cF.
\end{align}
\end{proposition}
\begin{proof}
The first equivalence follows from the fact of state fidelity that $F(\rho,\sigma)=1$ if and only if $\rho=\sigma$. The second equivalence follows since $\cF$ is closed by assumption.
\end{proof}



We now present error bounds for these channel error measures.  
For the Choi and average-case fidelities, note the following linearity property.
\begin{lemma}[Linearity]
Let $x \in \{\cho,\ave\}$ and $\cU$ be a unitary channel. Then $F_x(\cN, \cU)$ is linear in $\cN$. That is, given $\cN = p\cN_1 + (1-p)\cN_2$ for $p\in[0,1]$ and quantum channels $\cN_1,\cN_2$, it holds that
\begin{align}
    F_x(\cN, \cU) = pF_x(\cN_1, \cU) + (1-p)F_x(\cN_2, \cU).
\end{align}
\label{lem:linearity}
\end{lemma}

\begin{proof}
Consider the Choi fidelity first. We have 
\begin{align}
& F_{C}(p\cN_1 + (1-p)\cN_2, \cU) \notag\\
 &= F(\Phi_{p\cN_1+(1-p)\cN_2},\Phi_{\cU})\\
   &=\tr(\Phi_{p\cN_1+(1-p)\cN_2} \Phi_{\cU})  \\
   &= p\tr(\Phi_{\cN_1}\Phi_{\cU}) + (1-p)\tr(\Phi_{\cN_2}\Phi_{\cU})  \\
   &= p F_C(\cN_1 ,\cU) + (1-p)F_C(\cN_2 ,\cU),
\end{align}
where the second equality follows since the Choi state $\Phi_{\cU}$ is a pure state, and the third equality follows from the linearlity of the trace function.  Then due to Eq.~(\ref{eq: average case Choi state case equivalent}), we conclude that $F_A$ has the same linearity property.
\end{proof}
Collectively, our best bounds are the following.
\begin{theorem}\label{thm: nogo choi state case}
Given any quantum channel $\cN$ and any unitary target channel $\cU$,  it holds for any free superchannel $\Pi$ that 
\begin{align}
    \ve_\diamond(\cN\xrightarrow{\Pi}\cU)   &\geq  \ve_\wst(\cN\xrightarrow{\Pi}\cU) \nonumber\\ & \geq  \ve_\cho  (\cN\xrightarrow{\Pi}\cU)\geq \lfc_\cN(1-f^{\cho}_{\cU}),
    \label{eq:channel bound}
\end{align}
and 
\begin{align}
    &\ve_\ave(\cN\xrightarrow{\Pi}\cU) \geq \lfc_\cN(1-f^{\ave}_{\cU}) = \frac{d}{d+1}\lfc_\cN(1-f^{\cho}_{\cU}),
    \label{eq:channel bound ave}
\end{align}
where $d$ is the dimension of the input system of $\cU$.
\end{theorem}
\begin{proof}
The proof is analogous to that of Theorem~\ref{thm: nogo deterministic state transformation}.
By the definition of $\lfc_{\cN}$, there exists free channel $\cM \in \cF$ and channel $\cR$ such that $\cN$ can be decomposed as follows:
\begin{align}
    \cN = \lfc_\cN \cM + (1-\lfc_{\cN}) \cR.
\end{align}
Let $\Pi$ be any free superchannel. By the linearity of superchannels, 
\begin{align}
    \Pi(\cN) = \lfc_\cN \Pi(\cM) + (1-\lfc_{\cN}) \Pi(\cR).  \label{eq: channel decomposition}
\end{align}
 Then for $x\in\{\cho,\ave\}$, it holds that
 \begin{align}
     &\ve_x(\cN\xrightarrow{\Pi}\cU) \nonumber\\ &= 1 - F_{x}(\Pi(\cN),\cU)  \\
     &= 1 - F_{x}(\lfc_\cN \Pi(\cM) + (1-\lfc_{\cN}) \Pi(\cR),\cU) \\
     &= 1 - \lfc_\cN F_{x}\left(\Pi(\cM),\cU\right) - (1-\lfc_\cN) F_{x}\left(\Pi(\cR),\cU\right) \\
     &\geq 1 - \lfc_\cN f^{x}_{\cU} - (1-\lfc_{\cN})\\
     &= \lfc_\cN(1-f^{x}_{\cU}),
 \end{align}
 where the third line follows from the linearity property Lemma~\ref{lem:linearity}, and the inequality follows from the fact that $F_{\cho}(\Pi(\cM),\cU) \leq f^{\cho}_{\cU}$ since $\Pi(\cM)\in\cF$ by the golden rule, and $F_{\cho}\left(\Pi(\cR),\cU\right) \leq 1$.
Then by Eq.~(\ref{eq:ordering}) we obtain Eq.~(\ref{eq:channel bound}), and Eq.~(\ref{eq:channel bound ave}) follows from the relation Eq.~(\ref{eq: average case Choi state case equivalent}) and is essentially the same bound as the last one of Eq.~(\ref{eq:channel bound}) upto a dimension factor.  
\end{proof}
Note that the best bounds we can get for all error measures are in terms of the Choi overlap $f^{\cho}_{\cU}$. A natural question is whether one can directly use $f^{\wst}_{\cU}$ in the bound for $\ve_\wst$, which would improve the bound.  The problem is we do not have a linearity property analogous to Lemma~\ref{lem:linearity} for the worst-case fidelity $F_\wst$, so the third line does not go through.

As long as the target channel $\cU \not\in \cF$, it is clear by definition that $f_\cU^\cho < 1$.
That is, for any channel $\cN$ satisfying the $\lfc_\cN > 0$ condition and any resource unitary channel, all the above error bounds are nontrivial and thus imply a nonzero error.


\subsubsection{Multiple channel uses and adaptive channel simulation}
Now we discuss the scenario where one takes multiple noisy channels as inputs and intends to simulate some unitary channel, which is analogous to the standard task of distilling high-quality resources from many noisy resources in the state setting.  
However, the multiple instance setting represents a very important difference between channels and states. The composition of multiple states has a simple parallel structure represented by tensor products. In contrast, multiple channels can be used sequentially and adaptively, which is not simply described by tensor products and may be more powerful than the parallel scheme. Whether the adaptive scheme can outperform the parallel one is a crucial problem in many research areas about quantum channels, such as channel simulation, discrimination, and estimation (see, e.g.,~Refs.~\cite{Hayashi09:discrimination,PhysRevA.81.032339,pirandola2017fundamental,Pirandola19:discr,Pirandola20:programmable,wang2019resource,fang2020chain,WildeBertaHircheKaur,PhysRevLett.113.250801,PhysRevLett.118.100502,Yuan_2017,ChannelSimulationinQuantumMetrology,AdaptiveestimationanddiscriminationofHolevoWernerchannels,2020arXiv200310559Z,KatariyaWilde}).

First, note that the parallel use of multiple channels $\cN_1,\cdots, \cN_n$ is again represented by tensor product and thus can be simply regarded as a single channel $\cN = \bigotimes_{i=1}^n\cN_i $.  Therefore, the results  above can be directly applied.   In addition to error bounds, using the super-multiplicity property (Proposition~\ref{prop: lfc channel super multiplicity}), we directly bound the cost or overhead of unitary channel simulation, defined by the number of instances of a certain channel needed to simulate some unitary channel, using parallel strategies.
\begin{corollary}[Parallel simulation cost]
Suppose some free superchannel $\Pi$ transforms $n$ instances of noisy channels $\cN$ to target unitary channel $\cU$ with a certain type of error $\ve_x(\cN^{\ox n}\xrightarrow{\Pi}\cU) \leq \epsilon_x, x\in\{\diamond,\wst,\cho,\ave\}$. Then $n$ must satisfy:
\begin{align}
n & \geq \Bigg[\log\frac{1-f^\cho_\cU}{\epsilon_x}\Bigg]\Bigg[\log\frac{1}{\lfc_{\cN}}\Bigg]^{-1},
\end{align}
for any $x\in\{\diamond,\wst,\cho\}$.
The bound on $n$ in terms of the average-case error $\ve_A$ is equivalent to that in terms of the Choi error $\ve_C$.

\label{cor:overhead channel}
\end{corollary}


Now we consider the adaptive scheme, which represents a more general way to use multiple input channels to simulate an output channel.  Here, the action on input channels $\cN_1,\cdots, \cN_n$ is represented by a ``quantum comb'' \cite{Chiribella08:comb,Chiribella09:comb2} $\Pi$ with appropriate dimensions realized by channels $\cP_1,\cdots, \cP_{n+1}$, and the input channels are inserted into the slots (as depicted in Fig.~\ref{fig:comb}). 
In resource theory contexts, there is again a golden rule on the combs that a free comb must map free channels to a free channel, that is, if one inserts free  channels  $\cN_1,\cdots, \cN_n\in\cF$ in the slots of comb $\Pi_n$ then the overall channel $\Pi_n(\cN_{[n]})\in\cF$ (where $\cN_{[n]}$ is short for the channel collection $[\cN_1,\cdots, \cN_n]$). Note that, in the case where the comb is realized by free channels $\cP_1,\cdots, \cP_{n+1}\in\cF$ (and the identities on the ancilla systems are considered free), it obviously obey  the  golden rule, because axiomatically the composition of free  channels is free. However, the converse is not necessarily true, that is, the notion of free combs is more general than  free realization.
\begin{figure}[]
\centering
\includegraphics[]{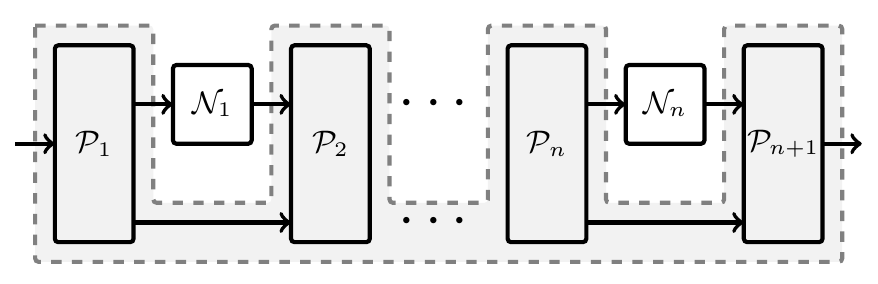}


\caption{Quantum comb. Given input channels $\cN_1,\cdots, \cN_n$, the general map that outputs a channel can be represented by a quantum comb (gray area) realized by channels $\cP_1,\cdots, \cP_{n+1}$, and the input channels are used by inserting them into the slots.
}\label{fig:comb}
\end{figure}





Note that the channel free component obeys the following monotonicity property under free combs.
\begin{proposition}[Monotonicity]\label{lem: lfc monotone under superchannel}
Given any channels $\cN_1,\cdots,\cN_n$ (collectively denoted by $\cN_{[n]}$), it holds that, for any free comb $\Pi_n$ acting on $\cN_{[n]}$, 
\begin{align}
\lfc_{\Pi_n(\cN_{[n]})} \geq \prod_{i=1}^n\lfc_{\cN_i}.    
\end{align}
\end{proposition}
\begin{proof}
Suppose the quantum comb $\Pi_n$ is realized by channels $\cP_i$ with $i=1,\cdots,n+1$, as depicted in Fig.~\ref{fig:comb}.
We emphasize that $\cP_i$ are not necessarily free channels themselves; the only requirement here is that the whole comb obeys the golden rule, i.e.~$\Pi_n(\cM_{[n]})\in\cF$ as long as $\cM_{[n]} = [\cM_1,\cdots,\cM_n]$ and $\cM_i \in \cF$ for $i \in [n]$.
Suppose that the maximization in $\lfc_{\cN_i}$ is achieved by $\cM_i \in \cF$, that is, $\cN_i \geq \lfc_{\cN_i} \cM_i, \forall i$. Then we have 
\begin{align}
\Pi_n(\cN_{[n]}) & = \cP_{n+1}\prod_{i=1}^n (\cN_i \ox \cI_i)\circ \cP_i\\
& \geq\cP_{n+1}\prod_{i=1}^n \lfc_{\cN_i} (\cM_i \ox \cI_i)\circ \cP_i \\
& = \prod_{i=1}^n\lfc_{\cN_i} \Pi_n(\cM_{[n]}), \label{eq:comb order}
\end{align}
where the inequality follows from the fact that channel tensorizations and compositions preserve the channel order $\cN_i \geq \lfc_{\cN_i} \cM_i$. Since $\Pi_n(\cM_{[n]}) \in \cF$ by the golden rule, we get $\lfc_{\Pi_n(\cN_{[n]})} \geq \prod_{i=1}^n\lfc_{\cN_i}$ by definition.
\end{proof}


Now for input channels $\cN_{[n]} = [\cN_1,\cdots,\cN_n]$, comb $\Pi_n$ and unitary target channel $\cU$, the simulation error is defined as
\begin{align}
    \ve_x(\cN_{[n]}\xrightarrow{\Pi_n}\cU) := \ve_x(\Pi_n(\cN_{[n]}), \cU),
\end{align}
for $x \in \{\diamond,\wst,\cho,\ave\}$.

By a little tweak of the proofs above, we establish bounds on the error and cost for adaptive simulation, which match those for the parallel case.
\begin{corollary}[Adaptive simulation error]\label{coro: Adaptive simulation error}
Given any channels $\cN_1,\cdots,\cN_n$ (collectively denoted by $\cN_{[n]}$ and any unitary target channel $\cU$), it holds that, for any free comb $\Pi_n$ acting on $\cN_{[n]}$,
\begin{align}
\ve_\diamond(\cN_{[n]}\xrightarrow{\Pi_n}\cU)& \geq \ve_\wst(\cN_{[n]}\xrightarrow{\Pi_n}\cU)\\
& \geq  \ve_\cho(\cN_{[n]}\xrightarrow{\Pi_n}\cU)\geq 
\prod_{i=1}^n\lfc_{\cN_i}(1-f^{\cho}_{\cU}), \notag
\end{align}
and
\begin{align}
\ve_\ave(\cN_{[n]}\xrightarrow{\Pi_n}\cU) & \geq  \prod_{i=1}^n\lfc_{\cN_i}(1-f^{\ave}_{\cU})\notag\\
& =\frac{d}{d+1}\prod_{i=1}^n\lfc_{\cN_i}(1-f_{\cU}^{\cho}),
\end{align}
where $d$ is the dimension of the input system of $\cU$.
\end{corollary}
\begin{proof}
Simply note that, according to Eq.~(\ref{eq:comb order}), we have the following decomposition
\begin{align}
\Pi_n(\cN_{[n]}) =  \prod_{i=1}^n\lfc_{\cN_i} \Pi_n(\cM_{[n]}) + \left(1-\prod_{i=1}^n\lfc_{\cN_i}\right)\cR, 
\end{align}
for some channel $\cR$.   
By following the arguments in the proof of Theorem~\ref{thm: nogo choi state case}, one  can establish similar error bounds where $\lfc_\cN$ is replaced by $\prod_{i=1}^n\lfc_{\cN_i}$.
\end{proof}
Therefore, we can establish the same bound on the simulation cost for the adaptive scheme.
\begin{corollary}[Adaptive simulation cost]\label{coro: Adaptive simulation cost}
Suppose some free  comb $\Pi_n$ transforms $n$ instances of noisy channels $\cN$ to target unitary channel $\cU$ with a certain type of error $\ve_x([\cN,\cdots,\cN]\xrightarrow{\Pi_n}\cU) \leq \epsilon_x, x\in\{\diamond,\wst,\cho,\ave\}$. Then $n$ must satisfy
\begin{align}
 n\geq \Bigg[\log\frac{1-f^\cho_\cU}{\epsilon_x}\Bigg]\Bigg[\log\frac{1}{\lfc_{\cN}}\Bigg]^{-1},
\end{align}
for any $x\in\{\diamond,\wst,\cho\}$. 
The bound on $n$ in terms of the average-case error $\ve_A$ is equivalent to that in terms of the Choi error $\ve_C$.
\label{cor:overhead channel adaptive}
\end{corollary}

Note that the adaptive strategies may potentially reduce the error or cost of simulation compared to parallel ones, so the adaptive simulation bounds can be regarded stronger.

A general observation is that the simulation cost asymptotically scales at least as $\Omega(\log(1/\epsilon_x))$ as target error $\epsilon_x \rightarrow 0$ even if we allow adaptive usages of the input channels, no  matter which kind of error measure $x$ is chosen. 



\subsubsection{No-purification conditions}

Here, we discuss  the situations  where no-go rules are in place  for channel resource purification, i.e.~no unitary resource channels can be exactly simulated.  For both the cases of single and multiple input channels, the basic statement goes  as follows. 
\begin{corollary}
There is no free superchannel (or comb) that exactly transforms channel $\cN$ (or a collection of channels $\{\cN_i\}$) to any unitary resource channel $\cU\notin\cF$ if $\lfc_\cN > 0$ (or $\lfc_{\cN_i} > 0,\; \forall i$).
\end{corollary}
\begin{proof}
Since $\cF$ is closed by assumption and $\cU \notin \cF$, we have $f_\cU < 1$. Then due to Theorem~\ref{thm: nogo choi state case} the transformation error (in whichever measure) is strictly positive, indicating that the exact transformation is impossible.
\end{proof}
Now similar to Proposition~\ref{prop:no-go conditions}, we give a series of alternative characterizations of the $\lfc>0$ condition for channels, which could be illustrative or useful in certain scenarios:
\begin{proposition}\label{prop:no-go conditions channel}
For any quantum channel $\cN$, the following conditions are equivalent: 
\begin{enumerate}
\item[(a)] Channel free component: $\lfc_\cN>0$;
\item[(b)] State free component:
\begin{enumerate}
    \item [(b1)] Worst case: For any input state $\rho$,  $\lfc_{\rho_\cN}>0$ (defined with respect to the set  of free states $\{\rho_{\cM}: \cM\in\cF\}$);
    \item [(b2)] Choi state:  $\lfc_{\Phi_\cN} > 0$ (defined with respect to the set of free states $\{\Phi_\cM: \cM\in\cF\}$);
\end{enumerate}
\item[(c)] Support:
\begin{enumerate}
    \item [(c1)] Worst case: There exists $\cM \in \cF$ such that, for any $\rho$,  $\mathrm{supp}(\rho_{\cN}) \supseteq \mathrm{supp}(\rho_{\cM})$;
    \item [(c2)] Choi state: There exists $\cM \in \cF$ such that  $\mathrm{supp}(\Phi_\cN) \supseteq \mathrm{supp}(\Phi_\cM)$;
\end{enumerate}
\item[(d)] Resource measure:
\begin{enumerate}
    \item [(d1)] Worst case: 
$ \min_{\cM\in\cF}{D}_{\min}(\cN\|\cM) = 0,$
where $D_{\min}(\cN\|\cM):=\sup_{\rho}{D}_{\min}(\rho_{\cN}\|\rho_{\cM})$ is the channel min-entropy between $\cN$ and $\cM$~\cite{CooneyMosonyiWilde,wang2019resource};
    \item [(d2)] Choi state:  $\min_{\cM\in\cF}{D}_{\min}(\Phi_\cN\|\Phi_\cM) = 0$.
\end{enumerate}
\end{enumerate}
\end{proposition}
\begin{proof}
First it is clear that (b1), (c1), (d1) are equivalent, and (b2), (c2), (d2) are equivalent, due to the state theory result, Proposition~\ref{prop:no-go conditions}. 
The equivalence between (b1), (b2), and (a) follows from the fact that $\cN \geq \gamma\cM$ is equivalent to $\Phi_\cN  \geq \gamma\Phi_\cM$ as well as $\rho_\cN \geq \gamma \rho_\cM$ for all $\rho$. 
\end{proof}
Worth noting, in the channel theory, the counterparts of  min-relative entropy monotones also nicely contrast noisy entities with pure ones.

\subsection{Practical scenarios and applications}\label{sec:channel applications}
The above no-purification rules and bounds are given in general forms so that their range of applicability is as wide as possible.  To provide some concrete understanding and guideline of their practical relevance, we now discuss some specific scenarios and applications of interest.
We shall start with a general discussion on typical noise models and the corresponding no-purification bounds in the contexts of different kinds of channel resource theories, and then specifically consider the roles of no-purification bounds in the contexts of quantum error correction, quantum communication, and circuit synthesis. Note that the main objective of our discussion here is to establish the frameworks for linking the no-purification principles to these practical problems. We shall mostly present general-form bounds, which are expected to be crude for certain specific resource features, noise models, system features etc., leaving refined analyses elsewhere.  

\subsubsection{Channel resource theories and practical noises}\label{sec:noises}


At a high level, we have the following two major different types of channel resource theories, signified by the role of the identity channel. 
\begin{itemize}
    \item {Information preservation theories.}
    In such theories, one is primarily interested in the noise channels and their abilities to simulate noiseless channels so as to preserve or transmit information.   Typical scenarios include quantum error correction and quantum communication. A signature of such theories is that the identity channel (between certain systems) is an ideal resource channel, representing no error or loss of quantum information occurring.   The set of free channels commonly involve e.g.,~certain constant (replacer) channels, which represent complete loss of information.   Here the free channels are in general directly induced by physical restrictions on the implementable operations that, e.g.,~perform the tasks of encoding and decoding.

    \item Resource generation theories. Such theories are commonly based on some resource theory defined at the level of  states (such as entanglement, coherence, magic states). The features of channels and simulation tasks of interest  are related to their ability of generating the state resource.  Here the set of free channels are derived from state theories and thus obey the resource  non-generating property (for  example, the identity channel is axiomatically free).  A typical scenario of this kind is synthesis, where a common task is to simulate, or ``synthesize'' some complicated target channel by elementary resource channels. See further discussions in the next part.
\end{itemize}
In some sense, theories of the first kind are intrinsically based on channels, and those of the second kind are induced by state theories.
Such classification may help elucidate the interplay  between channel and state resource theories.

Now we discuss typical noisy channels of interest in these two different kinds of channel resource theories.

First, consider the first kind, i.e.~information preservation theories, where the identity channel $\id$ is a resource.  Here,  the simulation capabilities (capacities) of noise channels themselves are of interest. A general observation is that, for stochastic noise 
\begin{align}
\cN_\mu = (1-\mu)\id + \mu\cN\;,
\end{align} 
where $\mu\in(0,1)$ is the noise rate, if the noise channel $\cN$ is considered free in the theory in consideration, then $\lfc_{\cN_\mu}\geq \mu$, which can  be directly used to establish bounds on simulation error and cost.   We list a few important noise models that are special cases:   (i) Depolarizing noise: $\cN(\rho) = I/d$ is just a constant channel that outputs the maximally mixed state;  (ii)  Erasure noise: $\cN(\rho) = \ketbra{{\perp}}{{\perp}}$ is also a constant channel that outputs an orthogonal garbage state;  (For these two cases $\cN$ is normally free as it essentially erases information completely.)   (iii)  Dephasing noise:  $\cN = \Delta$ which erases the off-diagonal entries and thus is typically free in quantum scenarios since all coherence-related information is lost;  (iv) Pauli noise: $\cN(\rho) = \sum_i\mu_i P_i\rho P_i$ where $\forall i~\mu_i\geq 0, \sum_i \mu_i = \mu$ and $P_i$'s are non-identity Pauli operators (note that this model encompasses the depolarizing and dephasing noises);  Here $\cN$ is a stabilizer operation, and thus the global Pauli noise has free component in the stabilizer theory, leading to limitations on stabilizer codes.   We shall demonstrate the connections to quantum error correction in more detail in Sec.~\ref{sec:qec}.  
Quantum communication is another important  scenario of this kind, which we shall discuss more specifically in Sec.~\ref{sec:capacities}.  

For the second kind, i.e.~resource generation theories, the input channels of practical interest are usually not the noise channels themselves but the resource-generating channels contaminated by noises. For example, consider $\cN_\mu\circ \cG = (1-\mu)\cG + \mu\cN\circ \cG$ where $\cN_\mu$ is a stochastic noise and $\cG$ is a noiseless resource-generating channel.  Also note that, in contrast to the first kind, the theory is commonly built upon a clear notion of free states. Then a general observation for this case is that if $\cN$ always output a free state, then $\lfc_{\cN_\mu\circ \cG}\geq \mu$. Again, this holds for the depolarizing and erasure noises in normal theories where the maximally mixed state and the garbage state are free (note that the bound can be loose in e.g.~magic theory; see Sec.~\ref{sec:synthesis}).  Then by definition, it also applies to dephasing noise in theories where the diagonal states are free (such as coherence and certain asymmetry theories).  
As mentioned, a particularly important problem in such theories is gate synthesis.
In Sec.~\ref{sec:synthesis}, we shall discuss the implications of our general results to practical synthesis problems in more detail.

Notably, certain communication problems and gate synthesis correspond to adaptive channel simulation, which cannot be understood in the single-channel or parallel simulation schemes.


\subsubsection{Quantum error correction}   \label{sec:qec}

    As a cornerstone of quantum computing and information \cite{nielsen2010quantum}, quantum error correction (QEC) serves to reduce noise effects and errors in physical systems by the idea of encoding the quantum information in a suitable way so that after noise and errors occur the original logical information can be restored (decoded) .    It is clearly important to understand various kinds of limits on QEC.
    Our results here are relevant to the broadly important scenario where the QEC procedures and codes obey certain rules or constraints.  Typical examples include  the well-studied stabilizer codes \cite{Gottesman97:thesis,nielsen2010quantum}, and covariant codes \cite{hayden2017error,faist2019continuous,woods2019continuous,kubica2020using,ZhouLiuJiang20,YangMo20:covariant,KongLiu}, which has recently drawn considerable interest in quantum computing and physics. {In Sec.~\ref{sec:state} we presented general limits on the QEC accuracy based on understanding the decoding as a purification task. 
    Here the channel framework provides an alternate formulation:   Notice that the QEC task is essentially to simulate an identity channel on the logical system; then the channel no-purification bounds induce fundamental limits on this channel simulation task.} As a result, we have the following general bounds on the accuracy and cost of constrained QEC when the system is subject to generic non-unitary noises ($L, S$ denote the logical and physical systems respectively):
\begin{corollary}[Constrained quantum error correction]
Suppose that the encoder and decoder are free channels (subject to certain resource theory constraints) $\Pi$. Then given noise channel $\cN_S$ acting on the physical system $S$, the commonly considered overall error measures for approximate QEC $\ve_x, x\in\{\diamond,\wst,\cho\}$ obey   
\begin{align}
\ve_x(\cN_S \xrightarrow{\Pi} \id_L) \geq {\lfc_{\cN_S}}(1- f_{\id_L}^{\cho}). \label{eq:constrained qec error}
\end{align}
For example, consider the natural independent noise model where the noise channel $\cN$ acts independently and uniformly on each subsystem (e.g.~qubit), i.e.,~the overall noise channel has  the  form $\cN_S = \cN^{\ox n}$. Then 
\begin{align}
\ve_x(\cN_S \xrightarrow{\Pi} \id_L)\geq ({\lfc_\cN})^n(1- f_{\id_L}^{\cho}),
\end{align}
and therefore, to achieve target error $\varepsilon_x({\cN_S}\rightarrow\id_L) \leq \epsilon_x$, the number of physical subsystems $n$ obeys
\begin{align}
n \geq \Bigg[\log\frac{1-f_{\id_L}^{\cho}}{\epsilon_x}\Bigg]\Bigg[\log\frac{1}{\lfc_\cN}\Bigg]^{-1}.
\end{align}
In  the  case  of stochastic noise $\cN = (1-\mu)\id + \mu\cM$ where 
$\cM\in\cF$, $\lfc_\cN$ in the above bounds can be replaced by $\mu$.
\end{corollary}

As previously noted, this general result applies to the important cases of stabilizer and covariant QEC, which, respectively, correspond to Clifford \cite{Gottesman97:thesis} and symmetry \cite{ZhouLiuJiang20} constraints.  
Note again that, in covariant QEC, under the commonly held assumption that the noise channel $\cN_S$ is covariant, we have the stronger conclusion that the error bound Eq.~(\ref{eq:constrained qec error}) holds for any decoder, meaning that covariant codes are no better than $[{\lfc_{\cN_S}}(1- f_{\id_L}^{\cho})]$-correctable for any decoder \cite[Lemma 2]{ZhouLiuJiang20}.
For independent Pauli and erasure noises in the stabilizer case, and depolarizing, dephasing, and erasure noises in the covariant case, $\lfc_\cN$ can be replaced by $\mu$ in the bounds. 

The bounds here are given in the most general forms, {indicating universal limitations on the accuracy and cost of constrained QEC schemes for any noise channel with free component like typical global noise channels, which are naturally important but underinvestigated in the context of QEC.  }   
It would be interesting to  perform more refined analysis of the bounds for specific constraints and noise models, which we leave for future work.

\subsubsection{Quantum communication and Shannon theory} \label{sec:capacities}






The central problem in quantum Shannon theory is to determine the capability of quantum channels to reliably transmit information. Depending on the purpose of transmission (e.g., transmitting classical or quantum information) and the resources that can be used at hand, there are many different variants of channel capacities, each of which corresponds to a channel simulation task in the language of resource theory (see, e.g.,~Refs.~\cite{pirandola2017fundamental,fang2019quantum,LiuWinter19,PhysRevLett.124.120502,wang2019converse,berta2018amortization,fang2019geometric}). 

Here we discuss quantum capacities, which correspond to the task of transforming a given channel to an identity channel between two distinct, distant parties (labs). Note that we need to distinguish the identity channel shared between distant labs from the local identity channel whose input and output systems belong to the same lab. The former is regarded as the ideal resource while the latter is completely free.   In resource theory language, channel capacities are determined by the choice of free superchannels or combs $\Pi$, which correspond to specific coding strategies. Some important cases include the following~\cite{wang2019converse,berta2018amortization,fang2019geometric}:
\begin{itemize}
    \item Unassisted code: superchannel $\Pi$ can be decomposed into an encoder $\cK_{A\to A'}$ by Alice composed with a decoder $\cD_{B\to B'}$ by Bob, i.e., $\Pi = \cD_{B \to B'} \cK_{A \to A'}$;
    \item Entanglement-assisted code: superchannel $\Pi$ acts as $\Pi(\cN)(\rho_A) = \cD_{B\bar B \to B'} \cN_{A'\to B} \cK_{A\bar A \to A'} (\rho_A \ox \omega_{\bar A \bar B})$ with encoder $\cK_{A\bar A \to A'}$, decoder $\cD_{B\bar B \to B'}$ and shared quantum state $\omega_{\bar A \bar B}$;
    \item Non-signalling assisted code: superchannel $\Pi$ is non-signalling from Alice and Bob and vice versa;
    \item  Two-way classical-communication-assisted code: quantum comb $\Pi$ can be realized by local operations and classical communication  (LOCC) operations $\cP_1,\cdots, \cP_{n+1}$ between Alice and Bob (see Fig.~\ref{fig:comb}).
\end{itemize}
Once the free superchannels or combs are set, the set of free channels is then implicitly defined as the channels that can be generated via these superchannels or combs. Note that the first three coding strategies correspond to parallel channel simulation while the last one corresponds to adaptive channel simulation.


The performance of quantum communication can be characterized by an achievable triplet $(n,k,\epsilon)$, meaning that there exists a $\Pi$-assisted coding strategy that uses $n$ instances of the resource channel to transmit $k$ qubits, or simulate $\id_{2^k}$ (identity channel on the system of dimension  $2^k$),  within $\epsilon$ error (here we consider Choi error, which is the standard choice of error measure for quantum communication).
Then by Corollary~\ref{coro: Adaptive simulation error}, we can obtain the following bounds on these parameters  for  general quantum communication in the non-asymptotic regime.  
\begin{corollary}[Quantum communication]\label{coro: capacity}
Suppose $(n,k,\epsilon_\cho)$ is an achievable quantum communication triplet by noise channel $\cN$ with an $\Pi$-assisted code. Then the Choi error $\epsilon_\cho$ obeys
\begin{align}\label{eq: channel capacity tmp1}
\epsilon_{\cho} \geq (\lfc_{\cN})^n(1-f^{\cho}_{\id_{2^k}}).
\end{align} 
In other words, the minimum number of channel uses required to enable reliable transmission of $k$ qubits within Choi error $\epsilon_\cho$ must satisfy
\begin{align}
n \geq \Bigg[\log\frac{1-f_{\id_{2^k}}^{\cho}}{\epsilon_\cho}\Bigg]\Bigg[\log\frac{1}{\lfc_\cN}\Bigg]^{-1}.
\end{align}
\end{corollary}

We now discuss in more detail the two-way assisted quantum capacity, which is of particular importance due to its close relation to the practical scenario of distributed quantum computing and quantum key distribution. Due to the notorious difficulty of adaptive communication strategies and the involved structure of LOCC operations, this quantum communication scenario is not well understood in spite of its practical importance.
The corresponding asymptotic setting that assumes infinite access to the resource channels was recently investigated by a relaxation of LOCC operations to the mathematically more tractable PPT operations (see, e.g.,~Refs.~\cite{berta2018amortization,fang2019geometric,fawzi2020defining}). In this case, we have the maximum overlap $f_{\id_{2^k}}^\cho  \leq \frac{1}{2^k}$~\cite{rains2001semidefinite}. As the quantum capacity concerns the maximum number of qubits that can be reliably transmitted per use of the channel, we can equivalently obtain from Corollary~\ref{coro: capacity} a nontrivial trade-off (which can be interpreted as a bound on the non-asymptotic two-way assisted quantum capacity):
\begin{align}\label{eq: capacity rate}
\frac{k}{n} \leq -\frac{1}{n} \log \left(1-\frac{\epsilon_\cho}{(\lfc_{\cN})^n}\right) \quad \text{if} \  \epsilon_\cho \leq (\lfc_{\cN})^n.
\end{align}
Also note that, since PPT operations are semidefinite representable, $\lfc_{\cN}$ here can be efficiently computed by 
\begin{align}
\lfc_{\cN} = \max \big\{& \tr W: \,\Phi_{\cN} \geq W,\, W \geq 0,\, \\ & W^{T_B} \geq 0,  \tr_B W_{AB} = (\tr W_{AB}) I_A/|A| \big\},\nonumber
\end{align}
where $\Phi_{\cN}$ is the Choi state of $\cN_{A\to B}$. Fitting this into Eq.~\eqref{eq: capacity rate} can help us do analysis beyond the asymptotic treatment and understand the intricate trade-off between different operational parameters of concern.

\subsubsection{Noisy circuit synthesis}  \label{sec:synthesis}

The problem of approximating some desired transformation by quantum circuits consisted of certain elementary gates, commonly studied under the name  of quantum circuit or gate or unitary synthesis (or sometimes known as ``compiling''), is crucial to the practical implementation  of quantum computation. 
Depending on the practical setting, it is often  the case that some gates are considered particularly costly as compared to other gates, and thus we are mostly interested in the amount of costly gates needed for the desired synthesis task.  
A key observation here is that such synthesis tasks can be formalized as adaptive channel simulation problems, where free gates form a comb and the costly gates are input channels that are inserted into the slots of the comb.  
A particularly important case is ``Clifford+$T$,'' where we would like to decompose the target transformation into Clifford gates, which are assumed to be free since they can be rather easily implemented fault tolerantly, and the ``expensive'' $T$ gates $T = \ketbra{0}{0} + e^{i\pi/4} \ketbra{1}{1}$. Note that the $T$ gates are often implemented by ``state injection'' gadgets \cite{GottesmanChuang99:injection} that make use of $T$ states produced by magic state distillation (studied in Sec.~\ref{sec:state}), which is a resource-intensive procedure.
Therefore, the key figure of merit we would like to optimize is the number of $T$ gates used (namely the ``$T$-count''); see, e.g.,~Refs.~\cite{PhysRevA.87.032332,kliuchnikov2013asymptotically,2013arXiv1308.4134G,Selinger,2015arXiv150404350K,babbush2018encoding,low2018trading} for a host of previous studies related to this problem.  Notably, resource theory is helpful for finding good bounds on the $T$-count in certain cases \cite{HowardCampbell17:magic_rt,bravyi2019simulation}.


{Existing literature on the synthesis problem mostly focuses on the noiseless scenario, where the elementary gates are unitary.  The noisy nature of practical (especially near-term) devices  motivates us to  consider the scenario where certain gates are intrinsically associated with noise and such noisy gates are the elementary components of the circuit for synthesis. For example, a key incentive for the Clifford+$T$ model is that the non-Clifford gates are much harder to protect compared to Clifford gates, so that one may want to consider intrinsically noisy non-Clifford gates  (see below).}  We note that there are fundamental differences between this noisy synthesis setting and the noiseless one, as seen later.  Now, the central question is how many noisy resource gates are needed to approximate a target unitary.
Based on the observation mentioned above which links the synthesis problem to adaptive channel simulation, we establish the following universal lower bounds on such ``noisy gate count'' from Corollary~\ref{coro: Adaptive simulation cost} (note that for synthesis problems we often use the diamond norm error).  
\begin{corollary}[Noisy gate count]
Consider the synthesis task of simulating unitary channel $\cU$ by channel (noisy gate) $\cG$ and arbitrary use of a set of free channels, which compose a free comb, within diamond norm error $\epsilon_\diamond$.  Then the number of instances of $\cG$ needed must satisfy
\begin{align}
n \geq \Bigg[\log\frac{1-f_{\cU}^{\cho}}{\epsilon_\diamond}\Bigg]\Bigg[\log\frac{1}{\lfc_{\cG}}\Bigg]^{-1}.
\end{align}
\label{cor: noisy gate count}
\end{corollary}

We now investigate the Clifford+$T$ case specifically, where the $T$ gate is associated with noise, and we are interested in the number of such noisy $T$ gates, or the ``noisy $T$-count''. 
Let $\cC_n = \{\cU_j\}_{j=1}^{N}$ be the $n$-qubit Clifford group consisting of $N$ discrete elements.  Let the set of free channels be the convex hull of $\cC_n$, i.e., $\cF = \conv (\cC_n)$, meaning that we allow mixtures of Clifford gates. Any free channel $\cM \in \cF$ can be written as a convex combination $\cM = \sum_{j=1}^N p_j \cU_j$ with $p_j \geq 0$ and $\sum_{j=1}^N p_j = 1$. For condition $\cN \geq \gamma \cM$, we can replace $q_j = \gamma p_j$ and obtain an equivalent condition $\cM \geq \sum_{j=1}^N q_j \cU_j$ with $q_j \geq 0$ and $\gamma = \sum_{j=1}^N q_j$. Therefore, the free component can be computed by a semidefinite program
\begin{align}\label{eq: SDP for clifford}
\lfc_{\cN} 
= \max \left\{\sum_{j=1}^N q_j: \Phi_\cN \geq \sum_{j=1}^N q_j \Phi_{\cU_j}, q_j \geq 0\right\},
\end{align}
where $\Phi_{\cN}$ and $\Phi_{\cU_j}$ are the Choi states of $\cN$ and $\cU_j$ respectively. 
As a concrete example, consider $T$ gate followed by depolarizing noise $\cN_\mu(\rho) = (1-\mu)\rho + \mu I/2$ as the elementary channel.   Its free component $\lfc_{\cN_\mu\circ T}$ is computed by the SDP Eq.~\eqref{eq: SDP for clifford} where $N = 24$ (see, e.g.,~Ref.~\cite{xia2015randomized} for an explicit enumeration of $\cC_1$), and depicted in Fig.~\ref{fig: noisy t}(a).   When $\mu \geq 1-\sqrt{3}/3 \approx 0.42$ we see that $\lfc_{\cN_\mu\circ T} = 1$, as $\cN_\mu$ compresses the entire Bloch sphere into the stabilizer octahedron so any output is a stabilizer state. Note that this explicit calculation improves the general bound $\mu$ as discussed in Sec.~\ref{sec:noises}.
In Fig.~\ref{fig: noisy t}(b), as an example, we plot the lower bounds on the noisy $T$-count in order to approximate a CCZ gate, obtained from Corollary~\ref{cor: noisy gate count} (where we used $f_{\CCZ}^\wst \leq 9/16$~\cite[Eq.(33)]{bravyi2019simulation}).  


\begin{figure}[]
\centering
\includegraphics[]{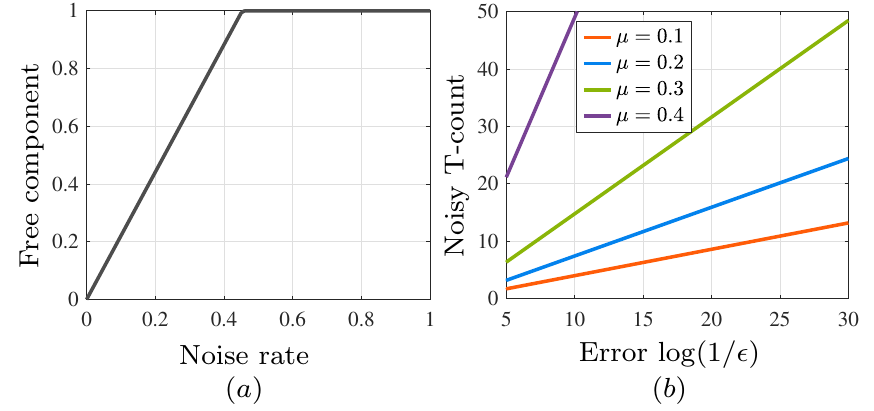}


\caption{Noisy Clifford+$T$ synthesis.  (a) Free component of the noisy $T$ gate $\cN_\mu\circ T$, where $T$ is followed by depolarizing noise of strength $\mu$. (b) Lower bounds on the noisy $T$-count (the number of $\cN_\mu\circ T$) needed to simulate a CCZ gate within diamond norm error $\epsilon$, depicted for noise rates $\mu = 0.1, 0.2, 0.3, 0.4$. 
}
\label{fig: noisy t}
\end{figure}





Recently, Ref.~\cite[Proposition 26]{wang2019quantifyingpublished} also gave an expression for the noisy gate counts in magic theory of odd dimensions using the mana monotone.   Note  that our result applies to any dimension, and is expected to outperform the mana bound especially in the small target error regime.  In particular, our bound implies diverging cost as the target error $\epsilon \rightarrow 0$, which is in line with intuitions, but the mana bound cannot.

Finally, we would like to 
remark that the noisy synthesis results here are  fundamentally different from the existing ones on noiseless synthesis, in spite of some apparent relations.
Most notably, it is known that for any universal gate set, the number of gates needed to approximate all unitaries up to error $\epsilon$ (which can essentially be measured by any channel error measure discussed earlier) scales at least as $\Omega(\log(1/\epsilon))$ \cite{Harrow02:lowerbound} (note that the well-known Solovay--Kitaev theorem \cite{Kitaev_1997,10.5555/2011679.2011685,nielsen2010quantum} concerns the upper bound).     Although the $\Omega(\log(1/\epsilon))$ scaling is similar to our lower bound on noisy gate counts, there are two key differences:  (i) Our noisy synthesis result bounds the number of resource gates needed and says nothing about the number of free gates, while the previous noiseless-case result counts the total number of gates; (ii) Our noisy synthesis result is universal for \emph{any} target resource unitary, while the previous noiseless-case result examines the worst case and there could well be target unitaries with lower or even trivial cost (some target unitaries can be exactly simulated, such as $T$ in Clifford+$T$).  
Relatedly, the geometric covering argument used in Ref.~\cite{Harrow02:lowerbound} is not useful for the noisy case.
In general, the noiseless and noisy synthesis and gate counts are fundamentally disparate problems contingent on different factors.  This can again be seen from  Clifford+$T$, where intricate number theory properties and techniques play decisive roles in the noiseless case \cite{PhysRevA.87.032332,kliuchnikov2013asymptotically,2013arXiv1308.4134G,Selinger,2015arXiv150404350K} while being irrelevant in the noisy case.


\section{Concluding remarks}\label{sec:conclusion}

We introduced a simple, universal framework for understanding and analyzing the limitations on quantum resource purification tasks that applies to virtually any resource theory, based on the notion of ``free component'' of noisy resources.  We developed the theory in detail for both quantum states and channels.   For the state theory, our new results significantly improve over corresponding ones discovered in Ref.~\cite{FangLiu20} in terms of both the regime of the no-purification rules and the quantitative limits.  This  framework also enabled us to  quantitatively understand the no-purification principles for quantum channels or dynamical resources.  Specifically, the channel theory involves complications concerning  error measures and the possibility of adaptively using multiple resource instances, as compared to the state theory.  We demonstrated broad theoretical and practical relevance of our techniques and results by discussing their applications to several key areas of quantum information science and physics.
 The simplicity and generality of our theory highlight the fundamental nature of the no-purification principles.

Several technical problems are worth further study.  First, we considered channel simulation with a single target channel here, but more generally the output can also be  a comb \cite{Chiribella08:comb};  It would be interesting to further study the no-purification bounds  for such cases   and  explore their relevance.  Second, we formulated the results in terms of deterministic one-shot transformation and only left preliminary remarks on the probabilistic case; A comprehensive understanding of the probabilistic case is left for future work.  Third, it is worth further study purification tasks for continuous variables, especially resource (e.g.~non-Gaussianity) distillation tasks and their applications in optical quantum information processing, given that there are some sharp distinctions known concerning the feasibility and behaviors of distillation procedures  \cite{lami18:gaussianrt,PhysRevA.97.062337} between continuous and discrete variables, but the understanding of the full correspondence is still preliminary.

Furthermore, it would be interesting to further analyze our bounds and associated parameters in specific theories and problems.
The discussion on the applications  we gave here mainly serve to establish the general, conceptual connections and are thus preliminary.  Further developments of these connections, taking specific features of the system, resource, and noise etc. into account,  may be fruitful.   In particular, for the extensively studied topics of quantum error correction and quantum Shannon theory, it would be interesting to further optimize the bounds and compare them with  existing results in specific scenarios.  We eventually hope that our demonstrations here will spark explorations of further applications or consequences of the no-purification principles in quantum information and physics.

\smallskip
\emph{Note added.}  After the completion of our paper, we became aware that Regula and Takagi independently considered the resource weight and obtained results related to ours which later developed into Ref.~\cite{RegulaTakagi:channel}. The two papers were arranged to be released concurrently on arXiv.

\begin{acknowledgments}
We thank Newton Cheng, Gilad Gour, Liang Jiang, Seth Lloyd, Milad Marvian, Kyungjoo Noh, and Sisi Zhou for discussions and feedback, Bartosz Regula and Ryuji Takagi for discussions and correspondence about their work \cite{RegulaTakagi:channel},  and especially Andreas Winter for inspiring discussions about the initial ideas.
ZWL is supported by Perimeter Institute for Theoretical Physics.
Research at Perimeter Institute is supported in part by the Government of Canada through the Department of Innovation, Science and Economic Development Canada and by the Province of Ontario through the Ministry of Colleges and Universities. 
\end{acknowledgments}

\bibliography{lfc_channel}

\end{document}